\documentclass[lettersize,journal]{IEEEtran}
\usepackage{amsmath,amsfonts}
\usepackage{algorithmic}
\usepackage{algorithm}
\usepackage{array}
\usepackage{subfigure}
\usepackage[caption=false,font=normalsize,labelfont=sf,textfont=sf]{subfig}
\usepackage{textcomp}
\usepackage{stfloats}
\usepackage{url}
\usepackage{verbatim}
\usepackage{graphicx}
\usepackage{cite}
\usepackage{amssymb}
\usepackage{amsthm}
\usepackage{multirow}
\usepackage{pifont}
\usepackage{circledsteps}

\newtheorem{proposition}{\textbf{Proposition}}

\makeatletter
\renewcommand{\maketag@@@}[1]{\hbox{\m@th\normalsize\normalfont#1}}%

\renewenvironment{proof}[1][\proofname]{\par\pushQED{\qed}
\normalfont \topsep6\p@\@plus6\p@\relax
\trivlist\item\relax
{\upshape#1\@addpunct{:}}\hspace\labelsep\ignorespaces}
{\popQED\endtrivlist\@endpefalse}

\makeatother

\begin{document}

\title{Power Allocation for Delay Optimization in Device-to-Device Networks: A Graph Reinforcement Learning Approach}
\author{Hao Fang, Kai Huang, Hao Ye, \IEEEmembership{Member,~IEEE}, Chongtao Guo, \IEEEmembership{Member,~IEEE}, \\Le Liang, \IEEEmembership{Member,~IEEE}, Xiao Li, \IEEEmembership{Member,~IEEE}, and Shi Jin, \IEEEmembership{Fellow,~IEEE}

\thanks{Hao Fang, Kai Huang, Xiao Li, and Shi Jin are with the National Mobile Communications Research Laboratory, Southeast University, Nanjing 210096, China (e-mail: fhao\_seu@seu.edu.cn; hkk@seu.edu.cn; li\_xiao@seu.edu.cn; jinshi@seu.edu.cn).

Chongtao Guo is with the College of Electronics and Information Engineering, Shenzhen University, Shenzhen 518060, China (e-mail: ctguo@szu.edu.cn).

Le Liang is with the National Mobile Communications Research Laboratory, Southeast University, Nanjing 210096, China, and also with Purple Mountain Laboratories, Nanjing 211111, China (e-mail: lliang@seu.edu.cn).

Hao Ye is with the Department of Electrical and Computer Engineering, University of California, Santa Cruz, CA 95064, USA (e-mail: yehao@ucsc.edu).}
}

\markboth{Journal of \LaTeX\ Class Files,~Vol.~14, No.~8, August~2021}%
{Shell \MakeLowercase{\textit{et al.}}: A Sample Article Using IEEEtran.cls for IEEE Journals}


\maketitle

\begin{abstract}
The pursuit of rate maximization in wireless communication frequently encounters substantial challenges associated with user fairness. This paper addresses these challenges by exploring a novel power allocation approach for delay optimization, utilizing graph neural networks (GNNs)-based reinforcement learning (RL) in device-to-device (D2D) communication. The proposed approach incorporates not only channel state information but also factors such as packet delay, the number of backlogged packets, and the number of transmitted packets into the components of the state information. We adopt a centralized RL method, where a central controller collects and processes the state information. The central controller functions as an agent trained using the proximal policy optimization (PPO) algorithm. To better utilize topology information in the communication network and enhance the generalization of the proposed method, we embed GNN layers into both the actor and critic networks of the PPO algorithm. This integration allows for efficient parameter updates of GNNs and enables the state information to be parameterized as a low-dimensional embedding, which is leveraged by the agent to optimize power allocation strategies. Simulation results demonstrate that the proposed method effectively reduces average delay while ensuring user fairness, outperforms baseline methods, and exhibits scalability and generalization capability.
\end{abstract}


\section{Introduction}
\IEEEPARstart{D}{evice-to-device} (D2D) communication, which enables the direct data exchange between devices without the involvement of base stations or relay devices, can occur both within and independently of cellular network coverage \cite{ref1}. This communication mode is particularly significant in 5G networks due to its potential to enhance communication efficiency, reduce delay, and increase network capacity \cite{ref2}. In next-generation wireless communication networks, including 6G and beyond, the demand for low-delay communication has become increasingly critical to support emerging applications and services \cite{ref3, ref4}. These include scenarios like autonomous driving, holographic communication, and extended reality, which impose extremely stringent reliability and delay requirements. Since 2017, the 3rd generation partnership project (3GPP) has been progressively developing 5G New Radio (NR) standards \cite{ref5}, which incorporate specific requirements and technical solutions to reduce end-to-end delay \cite{ref6}. However, since communication resources are inherently limited, it is essential to further explore how to ensure fair communication opportunities and service quality for users in D2D communication through effective radio resource management (RRM).


\subsection{Problem Statement and Motivation}

For an information packet in D2D communications, the delay refers to the duration from the moment the packet is generated at the transmitter to its reception by the receiver. This metric is crucial for evaluating the performance of communication networks, particularly for applications demanding real-time responsiveness. The main components of delay in D2D communication include: 1) \emph{Propagation delay}. The time it takes for a signal to travel from the transmitter to the receiver depends on the distance between the two devices and the signal propagation speed through the medium. In free space, electromagnetic waves travel at the speed of light, so propagation delay is usually minimal. 2) \emph{Processing delay}. The time required for processing data at both the transmitter and receiver, including tasks such as encoding, decoding, encryption, and decryption, is influenced by the processing power of the device and the complexity of the algorithm used. 3) \emph{Queueing delay}. The time that data spends in the queue at the device buffer for transmission is affected by network load and resource allocation policies. 4) \emph{Transmission delay}. The time it takes to push all the bits of the packet into the transmission medium depends on the size of the packet and the transmission rate of the link \cite{ref7}.


Since propagation delay and processing delay typically contribute minimally to overall delay, greater emphasis should be placed on queueing delay and transmission delay, which are directly influenced by system resource allocation, particularly power control in spectrum-sharing situation \cite{ref8}.
Moreover, transmit powers should be dynamically allocated according to the fast-changing network channel conditions in the D2D communication environment, which are influenced by user mobility and multi-path effects \cite{ref9}. However, accurately characterizing the intricate interference relationships among different users is often challenging, which poses significant difficulties in optimizing power allocation to coordinate different users effectively. More importantly, since different transmitters may have varying amounts of backlogged packets, the power allocation must be designed to fairly maximize the transmission rates of all users to reduce the network average delay \cite{ref10, ref11}. Therefore, it is essential to perform delay-aware power allocation that fully leverages interference topology and dynamically adapts to the network state.

Considering the significant capabilities of reinforcement learning (RL) in handling sequential decision-making problems and graph neural networks (GNNs) in capturing network topology features, we propose a graph reinforcement learning-based power allocation approach for D2D networks in this paper. While the proposed method adopts a centralized architecture, it remains applicable in practical small-scale scenarios, such as vehicle-to-infrastructure (V2I) networks, industry automation systems, and edge-computing-based Internet of things. The integration of GNNs with RL enhances the reflection of network states, thereby improving the performance of the optimization algorithm.

\subsection{Related Work}
Numerous studies have focused on RRM to improve transmission rate and reduce delay in wireless communications, such as from the perspectives of power control, spectrum band selection, and user selection. Spectrum sharing and power allocation in D2D-enabled vehicular networks were studied in \cite{addref11}, where slowly varying large-scale fading information is used to maximize V2I ergodic capacity while ensuring vehicle-to-vehicle (V2V) link reliability. Optimal resource allocation in a downlink non-orthogonal multiple access system with one base station and multiple users was investigated in \cite{ref12}, leading to a dynamic programming-based power allocation algorithm. This algorithm aims to maximize long-term network utility while satisfying practical power consumption constraints. To meet stringent reliability and energy efficiency requirements, a power control and rate allocation scheme was proposed in \cite{ref13} for single-input multiple-output wireless systems. This scheme maximizes system energy efficiency by leveraging only the average statistics of the signal and interference, along with the number of antennas available on the receiver side. In \cite{ref14}, spectrum and power were allocated to vehicular links using Markov models to maximize total ergodic capacity of V2I links while limiting the delay violation probability for V2V links.

Beyond traditional optimization methods, RL-based RRM approaches have been developed to reduce delay. To leverage RL for spectrum sharing and power allocation in vehicular networks, each V2V link was treated as an agent that independently determines the optimal sub-band and transmission power level based solely on local observations, thereby meeting stringent delay requirements and minimizing interference to V2I communications \cite{ref15}. The decentralized resource sharing problem in vehicular networks was modeled as a multi-agent RL problem in \cite{addref15}, where proper reward design and training mechanisms enable multiple V2V agents successfully to cooperate in a distributed way, enhancing the sum capacity of V2I links and reducing the delay of V2V links. In mobile edge computing (MEC) systems, RL has been leveraged to design RRM algorithms that coordinate the offloading of delay-sensitive computation tasks among multiple users. These strategies automatically optimize delay performance, providing offloading services with minimal delay to enhance service quality \cite{ref16}. In a single-queue and single-server wireless communication system with block fading channels, an efficient structural-optimistic Q-learning-based link scheduling algorithm was proposed in \cite{ref17}. This algorithm jointly considers queue length and channel conditions as states to minimize average delay while adhering to average energy consumption constraints. Despite significant advancements in RL-based RRM approaches, existing works often suffer from poor scalability, limited generalization, and a lack of interpretability, primarily due to the simplistic integration of neural network architectures. To overcome these limitations, it is essential to incorporate the specific structure of the target task into the design of neural networks.

We also pay attention to the research advancements in GNNs \cite{ref18, ref19, ref20} with applications in RRM \cite{ref21, ref22, ref23, ref24, ref25, ref26, ref27}. Unlike traditional neural networks (NNs) such as multi-layer perceptron (MLP) \cite{add1}, long short-term memory (LSTM) networks \cite{add2}, and Transformer \cite{add3}, which are primarily designed for Euclidean data (e.g., images, sequences), GNNs are specifically tailored to handle non-Euclidean spatial data, making them inherently suitable for modeling wireless communication networks \cite{add4} \cite{add5}. In D2D networks, where interference occurs between nearby devices, GNNs can effectively capture localized interference patterns by iteratively aggregating neighborhood information through message passing, accurately modeling the dynamic and sparse interference topology. In contrast, traditional NNs struggle to capture such spatial dependencies, as MLP process data independently, LSTM are designed for temporal dependencies rather than spatial interactions, and Transformer lack explicit graph representation, making them less effective for handling dynamic wireless topologies. In \cite{ref23}, a two-layer multi-layer perceptron was utilized to construct GNN layers for addressing large-scale RRM problems, with power control and beamforming as examples. The proposed method, trained in an unsupervised manner with unlabeled samples, can match or even outperform classical optimization-based algorithms, demonstrating high scalability. In \cite{ref26}, GNNs were employed to parameterize resource allocation policies to address resilient RRM optimization problems with per-user minimum capacity constraints. The approach adapts to network conditions using learnable slack variables, aiming to achieve a high sum rate while ensuring fairness among all receivers. Under limited iterations, traditional RRM algorithms that employ dual optimization methods are unable to ensure the feasibility of constraints. To address this issue, a state-augmented GNN-based RRM algorithm was proposed in \cite{ref27} to solve the power control problem in D2D networks. This algorithm combines the network states with the corresponding dual variables of the constraints as inputs, utilizing GNNs as the parameterized method for RRM strategies.


Few studies have explored the combination of RL and GNNs to enhance performance in wireless communication. In the wireless control system, where interference and fading patterns induce a time-varying graph among plants and controllers within the network, GNNs were employed in conjunction with primal-dual RL to train resource allocation and scheduling policies \cite{ref28}. To solve the problem of mutual interference between multiple devices in the internet of things, a spectrum allocation algorithm based on a graph convolutional network-double duel deep Q-Network was proposed in \cite{addref2801}, achieving higher throughput. Additionally, to optimize task offloading and minimize application delay in MEC systems, a graph reinforcement learning-based framework was developed in \cite{addref2802}. This framework models MEC as an acyclic graph, making offloading decisions through graph state migration, and has demonstrated significant advantages in reducing delay while exhibiting generalization to new environments and topologies. However, in these approaches, GNNs are primarily utilized as a pre-processing feature extractor with a cascaded structure, where gradients propagate in only one direction, and GNN parameters are not updated alongside the RL model. A critical aspect to consider in such frameworks is how to update the GNN parameters. To address this, a corresponding loss function needs to be designed, and the embeddings generated by GNNs require labeling, thereby forming a supervised learning framework.


Previous studies have primarily focused on maximizing user rate while giving less attention to average delay, which leads to persistent packet accumulation for low-rate users. This issue is particularly pronounced when maximizing the sum rate. To this end, we consider a more comprehensive set of factors in this paper, incorporating not only channel state information but also detailed buffer state information, including the number of transmitted packets, the number of backlogged packets, and the delay. The proposed centralized architecture, which integrates RL with GNNs, addresses the issue by capturing network topology and dynamically adapt to the network environment. RL provides the ability to learn and refine strategies based on continuous interaction with the environment, enabling adaptive and responsive optimization. Meanwhile, GNNs accurately model the network topology through message passing on graph structures. This combined approach provides a precise representation of both network topologies and states, leading to superior performance in resource allocation \cite{addref28, addref2803}. Moreover, our method can serve as an upper-bound performance benchmark by leveraging global network state information. This enables the model to optimize power allocation policies with full knowledge of system conditions, yielding near-optimal solutions.


\subsection{Contribution}
This paper seeks to investigate a power allocation approach for delay optimization in D2D networks by leveraging GNNs-based RL, with a focus on maintaining user fairness. Unlike conventional approaches where GNNs serve primarily as a pre-processing feature extractor in a cascaded structure, our method introduces a bidirectional embedding architecture that deeply integrates GNNs within the RL framework. Specifically, GNN layers are embedded directly into both the actor and critic networks within the proximal policy optimization (PPO) framework, enabling joint gradient backpropagation and parameter sharing, which enhances learning efficiency and improves network state representation. Moreover, our approach operates within a self-supervised RL framework, enabling the GNN-embedded agent to learn directly through interactions with the environment. This eliminates the need to label the emdeddings generated by GNNs, thereby enhancing adaptability to different network conditions. The main contributions of this work are summarized as follows.
%
\begin{itemize}
  \item We introduce a delay-aware state representation that not only incorporates channel state information but also considers packet delay, the number of backlogged packets, and the number of transmitted packets. This comprehensive representation facilitates more effective power allocation decisions, optimizing delay performance while ensuring user fairness.
  \item We propose a bidirectional embedding architecture that directly integrates GNNs into the actor and critic networks of the PPO algorithm, enhancing the utilization of topology information. This integration enables the state information to be parameterized as a low-dimensional embedding, which is crucial for the agent to effectively optimize power allocation strategies.
  \item We develop a self-supervised RL framework for power allocation, where the RL agent learns optimal policies through direct interactions with the environment This framework not only enables joint gradient backpropagation and parameter sharing but also enhance the scalability and generalization capability of the proposed method.
  \item We conduct extensive numerical evaluations on various D2D communication network configurations, demonstrating the superior performance of the proposed method compared to baseline approaches. Simulation results show that the proposed method significantly reduces average delay while ensuring user fairness. Furthermore, it exhibits superior scalability and generalization.
\end{itemize}

\subsection{Paper Organization}
The rest of the paper is organized as follows. The system model and problem formulation are presented in Section \ref{section 2}. The RL method based on GNN is proposed in Section \ref{section 3}, which provides a detailed description of the key components of the proposed method. Section \ref{section 4} provides our simulation results. The conclusion and future research is made in Section \ref{section 5}.

\section{System Model and Problem Formulation}\label{section 2}
We consider a wireless network, as illustrated in Fig. \ref{fig systemmodel}, consisting of a set of $M$ transmitters and a set of $M$ receivers. Each transmitter attempts to communicate with its associated receiver, forming $M$ D2D communication pairs, where all pairs share a single spectrum sub-band. Each transmitter has a data buffer, which is assumed to be large enough to avoid dropping packets. For the entire D2D network, the communication duration, during which statistics are gathered, is equally divided into $T_{sum}$ time slots, each with a length of $T$. This duration is on the order of hundreds of microseconds, during which the wireless channel is assumed to remain approximately static. We assume that packets can arrive at the buffer of each transmitter at any time. During each time slot, each transmitter sequentially sends a number of packets, determined by the channel capacity, to its associated receiver following the first-come-first-serve rule. Any remaining packets will accumulate in buffers and be sent later. 

\begin{figure*}[htbp]
\begin{center}
  \includegraphics[width=6.2in,height=3.3in]{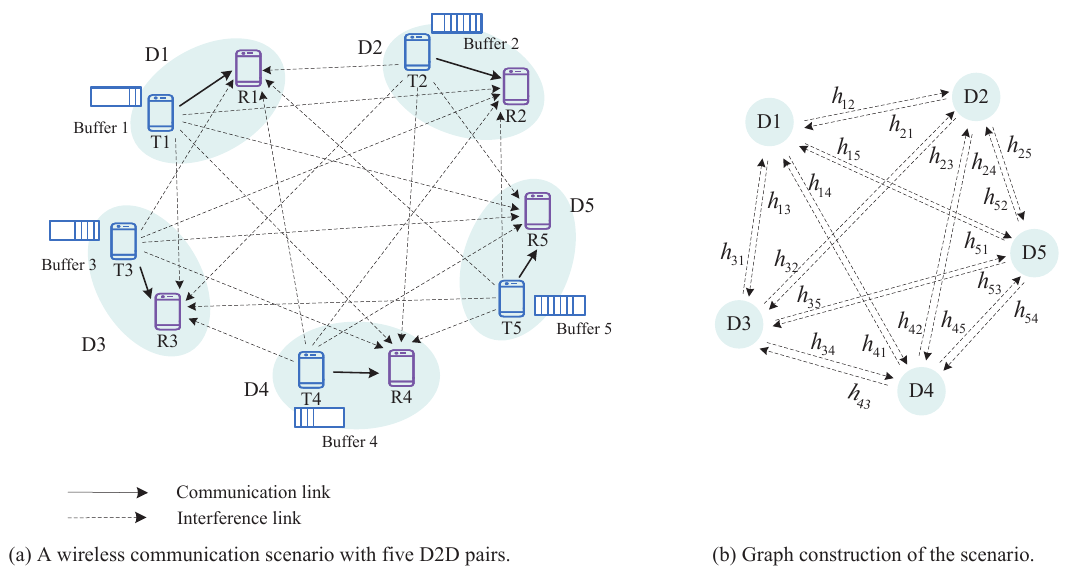}
  \caption{Illustration of a D2D communication scenario.}\label{fig systemmodel}
\end{center}
\end{figure*}

The channel gain from the transmitter Tx$_{i}$ to the receiver Rx$_{j}$, denoted as $h_{ij}$, is a random variable. The set of all channel gains forms the channel state, represented by ${\boldsymbol{\rm H}} \in {\mathbb{C}^{M \times M}}$. This channel state is typically composed of two components: long-term slow fading and short-term fast fading. We assume that the long-term slow fading remains constant throughout the communication duration, and only short-term fast fading, caused by multi-path propagation and user mobility, changes from one slot to another. Assuming all transmissions occur simultaneously on the same frequency band, they will interfere with each other. Then, the Shannon capacity of the $i$-th D2D link during the $n$-th time slot can be written as
\begin{equation}\label{adde2}
{\gamma _i}\left[ n \right] = B{\log _2}\left( {1 + \frac{{{{\left| {{h_{ii}}\left[ n \right]} \right|}^2}{p_i}\left[ n \right]}}{{{N} + \sum\limits_{j = 1,j \ne i}^M {{{\left| {{h_{ji}}\left[ n \right]} \right|}^2}{p_j}\left[ n \right]} }}} \right),
\end{equation}
where ${h_{ii}}\left[ n \right]$ denotes the direct channel coefficient for the $i$-th D2D pair in the slot $n$, ${h_{ji}}\left[ n \right]$ denotes the interfering channel coefficient formed from the $j$-th transmitter to the $i$-th receiver in the slot $n$, ${p_i}\left[ n \right]$ denotes the transmit power of the $i$-th transmitter in the slot $n$, $N$ is the noise variance, and $B$ is the bandwidth.

We assume that packets arrive at the transmitters for all D2D links following a Poisson arrival model. This assumption implies that, to decrease the average delay of the overall network, it is essential to allocate the powers of all transmitters judiciously to mitigate communication interference and ensure that each user can transmit packets efficiently. The number of packets arriving at the $i$-th transmitter during the $n$-th slot, ${\ell _i}\left( n \right)$, has a probability density function:

%
\begin{equation}\label{e1}
\Pr \left\{ {{\ell _i}\left( n \right) = k} \right\} = \frac{{{{\left( {{\lambda}T} \right)}^k}}}{{k!}}\exp \left( { - {\lambda}T} \right),{\rm{  }}k = 0,1,2, \cdots,
\end{equation}
where $\lambda$ is the average packet arrival rate. The number of packets that can theoretically be served by the $i$-th transmitter during the $n$-th time slot, denoted as ${D_i}\left[ n \right]$, has the probability mass function (PMF) given by
\begin{equation}\label{e2}
\Pr \left\{ {{D_i}\left[ n \right] = u} \right\} = \Pr \left\{ {\left\lfloor {\frac{{T{\gamma _i}\left[ n \right]}}{L}} \right\rfloor  = u} \right\},{\rm{    }}u = 0,1,2, \cdots,
\end{equation}
where $\left\lfloor  \cdot  \right\rfloor $ is the floor operation and $L$ is the number of bits per packet. The number of packets that the $i$-th transmitter can actually serve during the $n$-th time slot, denoted as ${\Xi_i}\left[ n \right]$, can be expressed as
\begin{equation}\label{e3}
{\Xi_i}\left[ n \right] = \min \left\{ {{D_i}\left[ n \right],{q_i}\left[ n \right]} \right\},
\end{equation}
where ${q_i}\left[ n \right]$ denotes the queue length at the beginning of the slot $n$. Thus, the queue in the buffer updates according to
\begin{equation}\label{e4}
{q_i}\left[ {n + 1} \right] = {q_i}\left[ n \right] - {\Xi_i}\left[ n \right] + {\ell _i}\left[ n \right].
\end{equation}

The delay experienced by each packet, denoted as $T_D$, consists of two components: queueing delay $T_b$ and transmission delay $T_s$. Queueing delay has a much greater impact on communication performance compared to transmission delay, while transmission delay is dependent on the allocated power. Therefore, we jointly consider both components of the delay. Let $m\left( i \right)$ denotes the number of packets successfully delivered by the $i$-th link during the communication duration. The delay of the $g$-th packet during $T_{sum}$ is represented as ${T_{D_i}(g)}$, which takes queueing delay and transmission delay into account, where $g \in \left\{ {1, \cdots ,m\left( i \right)} \right\}$. Our goal is to minimize the average delay for all D2D pairs, which can be expressed as
%
%
\begin{equation}\label{e5}
\begin{array}{l}
\mathop {\min }\limits_{\bf{p}} \frac{1}{M}\sum\limits_{i = 1}^M {\frac{1}{{m(i)}}} \sum\limits_{g = 1}^{m(i)} {{T_{{D_i}}}} (g)\\
{\rm{s}}{\rm{.t}}{\rm{. ~~~0}} \le {p_i} \le {P_{\max }}
\end{array}
\end{equation}
where $\textbf{p} = [{p_1},{p_2}, \cdots ,{p_M}]$ denotes the power vector of power to be optimized and $p_i$ denotes the power of the $i$-th transmitter, $i \in \{1, \cdots, M \}$.

Traditional approaches, along with some existing RL-based solutions, typically focus on optimizing the power allocation problem in \eqref{e5} based on the collected channel state ${\boldsymbol{\rm H}}$ and queue lengths, with users transmitting packets at the allocated rates. However, relying solely on these factors is insufficient to effectively reduce the average delay of the overall network. Hence, there is a pressing need to explore a more comprehensive set of factors to design the network state. Moreover, in some studies that utilize GNNs and RL simultaneously, GNNs are often employed for data embedding independently of RL, and the issue of GNN parameter updates has not been thoroughly addressed \cite{addref2804}. Therefore, it is crucial to give further consideration to how GNNs and RL can be integrated in a more coherent and effective manner.


\section{GNN-Based RL Method}\label{section 3}
In the D2D communication network scenario illustrated in Fig. \ref{fig systemmodel}, power allocation can be performed based on the collected channel state ${\boldsymbol{\rm H}}$ and the buffer state information from the previous time slot. This process is inherently a Markov decision process (MDP), making it well-suited to be modeled as an RL problem. We adopt a centralized RL approach, where a central controller collects and processes information about channel state ${\boldsymbol{\rm H}}$ and the buffers' states of all D2D pairs. The central controller acts as an agent, exploring the unknown communication environment to gain experience, which is subsequently utilized to guide policy design and refine power allocation strategies based on observations of the environment state. Since the power resource allocation problem may initially appear as a competitive game, we need to enable each user to achieve better data rates by selecting appropriate state information and designing a suitable reward function. This approach aims to minimize average delay while ensuring fairness among users. Key elements of the GNN-based RL method are described below in detail.

\subsection{State}

In the centralized RL framework for power resource allocation, a central controller operates as the agent, sequentially interacting with the environment to learn an optimal decision-making policy. This interaction can be formulated as an MDP, defined by the tuple $\left( {S,A,\rho,R} \right)$, where $S$ represents the state space, $A$ denotes the action space, $\rho$ is the state transition probability, and $R$ defines the reward function. As shown in Fig. \ref{PPO_GNN}, at each time slot $n$, the agent perceives the environment through the true system state $S_n$. Based on the observed information, the agent selects an action $a_n$ according to a policy $\pi \left( {\left. {{a_n}} \right|{S_n}} \right)$, where $\pi$ represents a probability distribution over possible actions. Following the action execution, the agent receives an immediate reward $R_{n}$, which quantifies the effectiveness of the chosen action, and the environment transitions to a new state $S_{n+1}$ according to the transition probability distribution $\rho\left( {\left. {S_{n+1}, R_{n}} \right|S_n, a_n} \right)$. Subsequently, the agent receives an updated state $S_{n+1}$, initiating the next decision cycle. Through repeated interactions, the agent aims to optimize its policy $\pi$ to maximize the expected cumulative reward over time. This is achieved by leveraging reinforcement learning algorithms that iteratively refine the policy based on observed rewards and state transitions, thereby enabling efficient power allocation strategies in dynamic wireless communication environments.

\begin{figure*}[htbp]
\begin{center}
  \includegraphics[width=6.3in,height=3.5in]{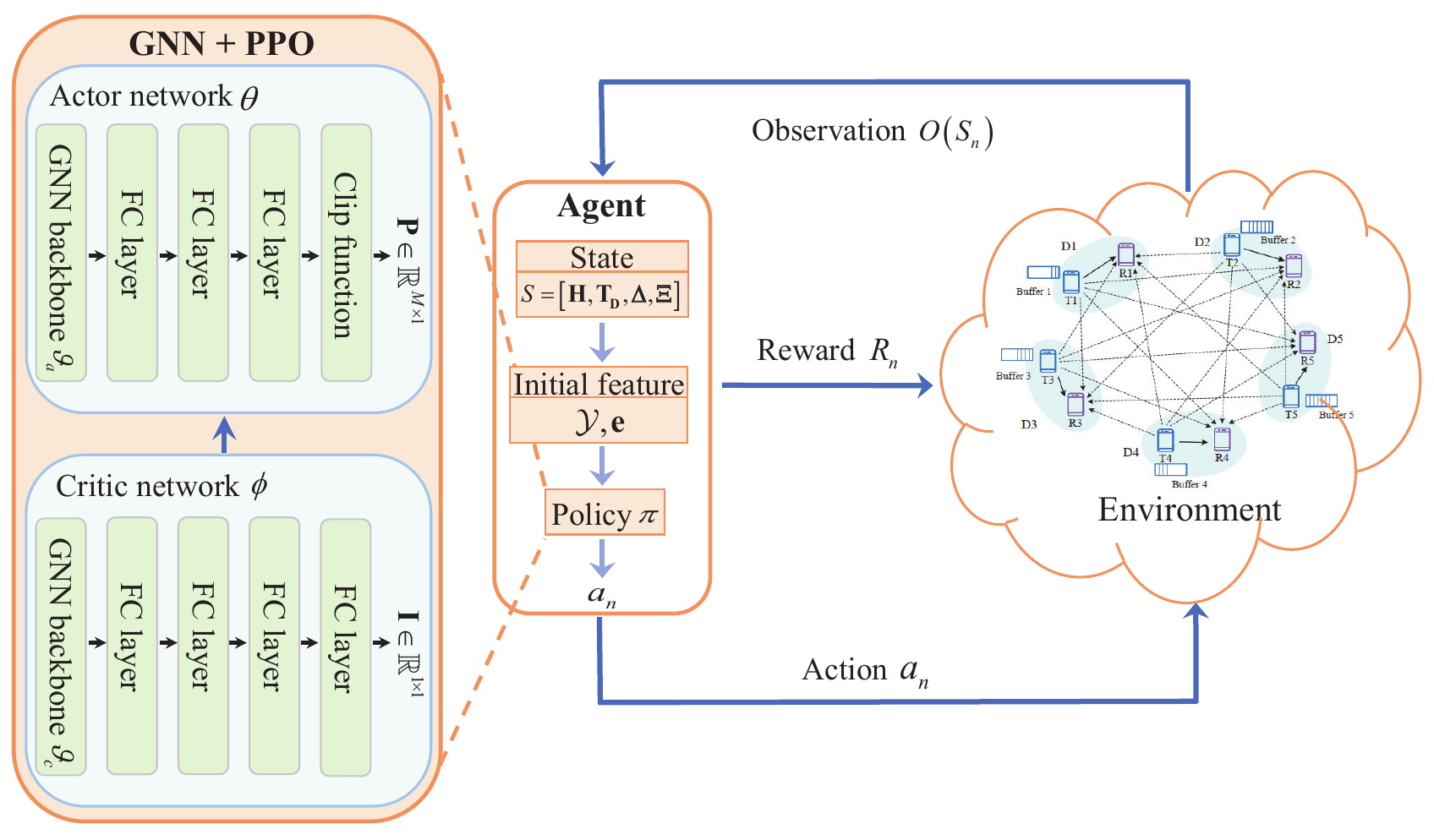}
  \caption{The overall process of the proposed method.}\label{PPO_GNN}
\end{center}
\end{figure*}

The channel quality significantly impacts communication performance in wireless communication, making channel gain critical in power allocation. Additionally, factors about the buffer state information, such as the number of packets in the buffer for each D2D pair and their experienced delay, are crucial for the purpose of optimizing delays in the considered D2D networks. Therefore, channel states ${\boldsymbol{\rm H}} \in {\mathbb{C}^{M \times M}}$, the delay vector of transmitted packets in all D2D buffer ${\boldsymbol{\rm{T_D}}} \in {\mathbb{R}^{M \times 1}}$, the number of backlogged packets ${\boldsymbol{\rm{\Delta}}} \in {\mathbb{R}^{M \times 1}}$, and the number of transmitted packets ${\boldsymbol{\rm \Xi}} \in {\mathbb{R}^{M \times 1}}$ within each time slot are selected as the environment state $S$. The state $S_n$ across the D2D network in time slot $n$, composed of the four elements: ${\boldsymbol{\rm H}_n}$, ${\boldsymbol{\rm{T_D}}_n}$, ${\boldsymbol{\rm{\Delta}}_n}$, and ${\boldsymbol{\rm \Xi}_n}$, can be summarized as follows.
\begin{equation}\label{e14}
S_n = \left\{ {\boldsymbol{\rm H}}_n, \boldsymbol{\rm{T_D}}_n, \boldsymbol{\rm{\Delta}}_n, \boldsymbol{\rm {\Xi}}_n \right\}.
\end{equation}

\subsection{Action Space}
The average delay minimization problem in \eqref{e5} of D2D communication networks essentially comes down to the transmission power control problem for D2D pairs. We consider the problem of continuous power resource allocation, where the allocated power levels have to satisfy the maximum power constraint ${p} \in \left[ {0,{P_{\max }}} \right]$. It is important to note that the dimension of the action space is directly related to the number of D2D pairs, i.e., $ \textbf{p}\in\mathbb{R}^{{M \times 1}}$. 

\subsection{Reward Design}
According to \eqref{e5}, a natural intuition is to use the average delay of packets transmitted up to the current moment as the reward. However, the direct delay-based reward function may lead to unintended consequences. Since the objective is to minimize the average delay across all successfully transmitted packets, directly using the average delay of packets transmitted up to the current moment as the reward could drive the RL agent toward an extreme policy where transmission power is minimized for all users, preventing packet transmissions altogether. In such a scenario, the measured delay would be zero, which is mathematically optimal but practically meaningless.

From the perspective of each data packet, once it enters the buffer, if it can be transmitted immediately without any retention, its delay is simply the transmission time required under the current allocated power. However, if there are packets ahead of it in the buffer, it must wait and its delay will increase by one for each additional time slot during which it remains in the buffer, which corresponds to the cumulative reward. Therefore, we count the number of packets in all buffers at each time step, allowing the cumulative reward to reflect the queueing delay experienced accurately. The reward function is designed as the negative value of the number of packets accumulated in buffers \cite{ref13} at each time step:
\begin{equation}\label{e16}
{R_n} = -\sum\limits_{i=1}^M{{\cal J}_i\left[ n \right] },
\end{equation}
where ${\cal J}_i\left[ n \right]$ denotes the number of packets accumulated in buffer $i$ at the time slot $n$. By defining the reward in terms of packet accumulation, we encourage the agent to allocate sufficient transmission power, ensuring that packets are transmitted in a timely manner while minimizing long-term queueing delay. This formulation maintains a balance between efficient transmission and delay minimization, leading to a more practical and effective resource allocation strategy.

It is widely acknowledged that an effective reward function should align its cumulative rewards as closely as possible with the objective function to be optimized. This alignment ensures that the RL method can effectively approximate the optimal solution of the objective function. The reward function we designed in \eqref{e16} is tailored to closely correspond with the objective function described in \eqref{e5}. Accordingly, the following proposition is stated.
\begin{proposition}
Given that the number of D2D pairs, $M$, and the average packet arrival rate, $\lambda$, are constant, maximizing the expected return in the MDP framework becomes asymptotically equivalent to minimizing the average delay across all D2D pairs in the original problem. This equivalence can be expressed as

\begin{equation}
\max \sum\limits_{n = 1}^{{T_{sum}}} {{R_n}}  \Leftrightarrow \mathop {\min }\limits_{\mathbf{p}} \frac{1}{M}\sum\limits_{i = 1}^M {\frac{1}{{m(i)}}} \sum\limits_{g = 1}^{m(i)} {{T_{{D_i}}}} (g),
\end{equation}
where $\left( \cdot \right)$ $\Leftrightarrow$ $\left( \cdot \right)$ denotes that the former is equivalent to the latter.
\end{proposition}

\begin{proof}[\rm{\textbf{Proof}}]
 Please see the Appendix A.
\end{proof}

\subsection{Learning Algorithm}
Focusing on an episodic setting, each episode covers the communication duration. In each time slot, the agent selects actions, i.e., allocated power levels $\textbf{p}$, based on the current state $S$. Following selecting the actions, the number of backlogged packets in each buffer changes, leading to an update in the reward function. This process is continuously optimized through an RL algorithm to minimize the number of backlogged packets in the buffers, thereby reducing delay.

The use of a graph structure is particularly suitable for wireless communication networks due to its ability to accurately model the inherent characteristics of such systems. In D2D networks, interference typically occurs only among nearby devices, resulting in a sparse topology. GNNs effectively capture these localized interference patterns, enhancing the accuracy of network modeling. Moreover, GNNs leverage local message passing, where each device aggregates information from its neighbors, naturally aligning with the localized nature of interference propagation. This message-passing mechanism allows GNNs to iteratively aggregate and propagate information across the network, capturing complex dependencies and patterns inherent in graph structures. Additionally, the graph representation is highly adaptable to dynamic topologies, allowing the model to effectively handle changes caused by device mobility and channel variations, thereby improving its robustness and generalization capabilities.

Unlike traditional RL methods such as deep Q-networks, which require retraining resource allocation strategies when the size of  communication network changes, the proposed method embeds GNN layers into the RL framework, enabling efficient generalization across different network sizes without the need for retraining. By leveraging GNNs to parameterize channel and buffer-related information into low-dimensional embeddings, this approach not only reduces the computational overhead associated with network variations but also exhibits superior scalability and generalization.

In the D2D communication network shown in Fig. 1(a), the links can be specifically divided into two types:
\begin{itemize}
\item \emph{Communication link}: In each D2D pair, there is a transmitter Tx$_i$, $i \in \left\{ {1, \cdots ,M} \right\}$ serving a receiver Rx$_i$. Communication links are contained in the nodes and are not characterized.
\item \emph{Interference link}: In different D2D pairs, the active transmitter Tx$_i$ in one D2D pair may interfere with the reception of receivers Tx$_j$, $j \ne i$, $j \in \left\{ {1, \cdots, M} \right\}$, belonging to the other D2D pairs. The channel of the interfering link is characterized as the connecting edge between different D2D pairs.
\end{itemize}

To leverage GNNs effectively, we begin by modeling the D2D network as a directed graph $\mathbb{G} = \left( {\mathbb{V},\mathbb{E},\mathbb{W}} \right)$, where $\mathbb{V}$ denotes the set of vertices that correspond to D2D pairs. Notably, the direct communication link for each D2D pair is contained in the vertices $\mathbb{V}$. The $\mathbb{E}$ denotes the set of interference links between different D2D pairs, where transmitters Tx$_i$ may cause interference to receivers Rx$_j$, $j \ne i$, $j \in \left\{ {1, \cdots ,M} \right\}$. These links are further characterized as edge weights determined by $\mathbb{W}$. To make better use of the topological information of the wireless network, we adopt the GNN architecture based on the graph $\mathbb{G}$. In the following, we describe the key elements of GNNs, including the initial node feature and edge feature, as well as how such features are processed and aggregated through multiple layers over the graph. Therefore, the state $S = \left\{ {\boldsymbol{\rm H}}, \boldsymbol{\rm{T_D}}, \boldsymbol{\rm{\Delta}}, \boldsymbol{\rm {\Xi}} \right\}$ needs to be processed as the initial node feature and edge feature, which will be discussed in detail below.

We introduce the proportional fairness (PF) ratio \cite{ref26} to ensure fair resource allocation decisions. Specifically, at each time slot $n$, for the receiver Rx$_i$, $i \in \left\{ {1,2, \cdots , M} \right\}$, the initial node feature includes the proportional-fairness ratio ${\Omega}{_i}\left[ n \right]$ defined as
\begin{equation}\label{e6}
  {\Omega}{_i}\left[ n \right] = \frac{{{{\hat \gamma}_i}\left[ n \right]}}{{{{\bar \gamma}_i}\left[ n \right]}},
\end{equation}
where ${\hat \gamma}{_i}\left[ n \right]$ denotes the estimated rate which is defined as
\begin{equation}\label{e7}
  {\hat \gamma_i}\left[ n \right] = {\log _2}\left( {1 + \frac{{{{\left| {{h_{ii}}\left[ n \right]} \right|}^2}{P_{\max }}}}{{{N} + \sum\limits_{j = 1,j \ne i}^M {{{\left| {{h_{ji}}\left[ n \right]} \right|}^2}{P_{\max}}} }}} \right),
\end{equation}
 and ${{{\bar \gamma}_i}\left( t \right)}$ denotes the exponential moving-average rate, which can be obtained through the following recursive update:
\begin{equation}\label{e8}
{\bar \gamma_i}\left[ n \right] = \left( {1 - \varepsilon } \right){\bar \gamma_i}\left[ {n - 1} \right] + \varepsilon {\gamma_i}\left[ n \right],
\end{equation}
with ${\gamma_i}\left[ n \right]$ being the actual achieved rate of the receiver Rx$_i$ at time slot $n$, and $\varepsilon  \in \left[ {0,1} \right]$ is the inverse averaging window length.

For the delay calculation, we consider the time elapsed from the moment a packet enters the buffer until it leaves the buffer. Specifically, we consider both the waiting time and the transmitting time of the backlogged packets in each buffer at time slot $n$. For the D2D pair $i$, the node feature characterizing the delay includes the following normalized delay:
\begin{equation}\label{e9}
{{\cal T}_i}\left[ n \right] = \frac{{T_{D_i}\left[ n \right]}}{{\sqrt {\sum\limits_{j = 1}^M {{{\left( {T_{D_j}\left[ n \right]} \right)}^2}} } }}.
\end{equation}

To count for the number of packets backlogged in each buffer at time slot $n$, for the D2D pair $i$, the node feature about the number of backlogged packets ${{\cal B}_i}\left[ n \right]$ is defined as
\begin{equation}\label{e10}
{{\cal B}_i}\left[ n \right] = \frac{{{\Delta _i}\left[ n \right]}}{{\sqrt {\sum\limits_{j = 1}^M {{{\left( {{\Delta _i}\left[ n \right]} \right)}^2}} } }},
\end{equation}
where ${{\Delta _i}\left[ n \right]}$ denotes the number of packets backlogged in buffer $i$ at the time slot $n$.

For the number of transmitted packets within each time slot, the node feature ${\boldsymbol{ \varpi }_i}\left[ n \right]$ is defined as
\begin{equation}\label{e11}
{\boldsymbol{ \varpi }_i}\left[ n \right] = \frac{{{\Xi_i}\left[ n \right]}}{{\sqrt {\sum\limits_{j = 1}^M {{{\left( {{\Xi_i}\left[ n \right]} \right)}^2}} } }},
\end{equation}

Therefore, the initial node feature is characterized by the above four components of information, that is, ${\cal {Y}} \buildrel \Delta \over = \left[ { \Omega, {\cal T}, {\cal B}, \boldsymbol{\varpi}} \right]$.


We use the channel gain as the edge feature, denoted by ${\boldsymbol{\rm{e}}} \in {\mathbb{R}^{M \times M}}$, which is normalized in dB based on \cite{ref26}. For the interference channel formed between the transmitter of the $i$-th D2D pair and the receiver of the $j$-th D2D pair in the slot $n$, the edge feature ${e_{ij}}\left[ n \right]$ is represented as follows:
\begin{equation}\label{e12}
{e_{ij}}\left[ n \right] = \frac{{\log \left( {\frac{{{P_{\max }}}}{{{\sigma ^2}}}{{\left| {{h_{ij}}\left[ n \right]} \right|}^2}} \right)}}{{\sqrt {\sum\limits_{i' = 1}^M {\sum\limits_{j' = 1}^M {{{\left[ {\log \left( {\frac{{{P_{\max }}}}{{{\sigma ^2}}}{{\left| {{h_{i'j'}}\left[ n \right]} \right|}^2}} \right)} \right]}^2}} } } }}.
\end{equation}

For each node $j \in \mathbb{V}$, the initial feature vector is denoted as ${\cal Y}_j^{\left( 0 \right)}$. Assuming that the architecture has $K$ GNN layers,  each node feature is updated through $K$ rounds of message passing and aggregation via the corresponding edge of graph $\mathbb{G}$. We use the GNN backbone proposed in \cite{ref29}, and its $k$-th layer output of node $j$ can be represented as
\begin{equation}\label{e13}
\begin{array}{l}
{{\cal Y}_j}^{\left( k \right)} =  \alpha  \bigg({{\cal Y}_j}^{\left( k-1 \right)}\vartheta _1^{\left( k \right)} + \\
~~~~~~~~~~~~~ ~~\sum\limits_{i \in {{\cal N}_j}: \left( {i,j} \right) \in \mathbb{E}} {e_{ij} \left({{{\cal Y}_j}^{\left( k-1 \right)}\vartheta _2^{\left( k \right)} - {{\cal Y}_i}^{\left( k-1 \right)}\vartheta _3^{\left( k \right)}} \right)} \bigg),
\end{array}
\end{equation}
%
%
where ${\cal N}_j$ denotes the set of neighbor nodes of node $j$, $\vartheta = \left( {\vartheta _1^k,\vartheta _2^k,\vartheta _3^k} \right)$ are learnable parameters and all in $\mathbb{R}^{{F_{k - 1}} \times {F_k}}$. $F_k$ is the dimension size of the $k$-th GNN layer. $\alpha  \left(  \cdot  \right)$ is the LeakyReLU nonlinear activation function.

With a multi-layer GNN architecture described in \eqref{e13}, each node can obtain information passed by its multi-hop neighbors. After $K$ layers, each node forms the final feature vector, that is, the node embedding ${{\rm E}_j}: = {\boldsymbol{\cal Y}}_j^{\left( K \right)} \in {\mathbb{R}^{{F_L}}}$. The embedding is then fed into the RL algorithm, where the agent learns and optimizes the policy through continuous interactions with the environment, ultimately obtaining the power control strategy. Since the GNN layer is integrated into the RL framework, its parameters are jointly updated along with the RL parameters during the training process.

\begin{figure}[htbp]
\begin{center}
  \includegraphics[width=3.4in,height=2.8in]{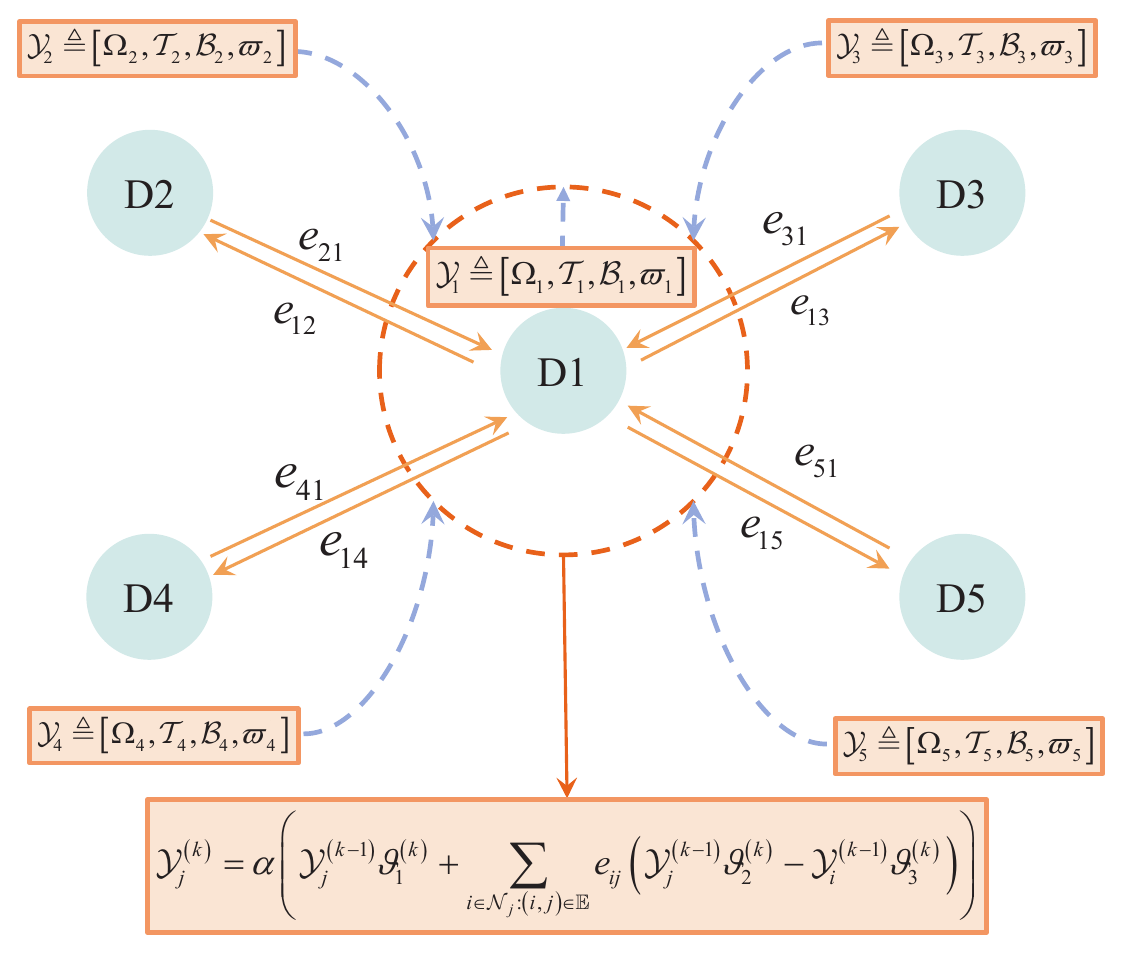}
  \caption{Illustration of implementing GNN architecture with five D2D pairs.}\label{GNN_backbone}
\end{center}
\end{figure}

The trust region optimization in the PPO algorithm \cite{ref30} ensures reliable and consistent policy updates, while its capability to handle continuous action spaces makes it highly effective for the power allocation scenario. Therefore, we adopt the PPO algorithm to address the continuous power control problem and train the agent efficiently. The agent has an actor network to choose the action and a critic network to evaluate the value of the current state. The allocated power levels satisfy the power constraint ${p_i}[n] \in \left[ {0,{P_{\max }}} \right]$ and follow a multivariate Gaussian distribution, where the mean and standard deviation are estimated by a neural network, so the last layer of the actor network is as follows.
%
%
%
\begin{equation}\label{e15}
\begin{gathered}
  \hat {\rho} \sim \mathcal{N}(\frac{{{P_{\max }}}}{2}{\rm{tanh}}(\mu ({{\mathbf{Y}}})) + \frac{{{P_{\max }}}}{2},\\
   ~~~~~~~~~~~~{\rm{softplus}}(\sigma ({{\mathbf{Y}}})) + 0.001), \hfill \\
\end{gathered}
\end{equation}
where ${\bf{Y}}$ is the output of the penultimate layer of the actor network. $\mu \left(  \cdot  \right)$ and $\sigma \left(  \cdot  \right)$ are a layer of neural networks that are used to generate the mean and variance of the expected power $\hat {\textbf{p}}$, respectively. $\rm{tanh}\left(  \cdot  \right)$ and $\rm{softplus}\left(  \cdot  \right)$ are two nonlinear activation functions. The expected power $\hat {\textbf{p}}$ is sampled based on the probability $\hat {\rho}$ and then the action, i.e. $ \textbf{p}\in\mathbb{R}^{{M \times 1}}$, can be expressed as
\begin{equation}\label{adde15}
\textbf{p}= {\rm{clip}}\left( {{ \hat {\textbf{p}}},0,P_{\max}} \right),
\end{equation}
where ${\rm{clip}}(x, \rm{min}, \rm{max})$ limits the input $x$ to the interval $\left[ {\min ,\max } \right]$ as an output. For the PPO algorithm, the objective function of training the actor network is
\begin{equation}\label{e17}
{\mathcal{L}^a}\left( \theta  \right) = {\mathbb{E}_t}\!\left[ {\min \left( {{r_t}\left( \theta  \right){{\hat A}_t},{\rm {clip}}\left( {{r_t}\left( \theta  \right),1 \!-\! \epsilon ,1\! +\! \epsilon } \right){{\hat A}_t}} \right)} \right],
\end{equation}
%
%
%
where $\theta$ is the trainable parameter of the policy network, ${r_t}\left( \theta  \right) = {\pi _\theta }\left( {\left. {{a_t}} \right|{s_t}} \right)/{\pi _{{\theta _{old}}}}\left( {\left. {{a_t}} \right|{s_t}} \right)$ is the probability ratio of the new policy $\pi _\theta$ to the old policy $\pi _{{\theta _{old}}}$, which is the key element of importance sampling. The second term inside the min operation, ${{\rm {clip}}\left( {{r_t}\left( \theta  \right),1 - \epsilon,1 + \epsilon} \right){{\hat A}_t}}$, is used to limit the magnitude of gradient updates to stabilize the training process, and $\epsilon$ is a hyperparameter. ${\hat A}_t$ is the generalized advantage estimation (GAE) function, which is utilized to evaluate the quality of an action in comparison with the average action taken in that state. The objective function for training the critic network can be expressed as
\begin{equation}\label{e18}
{\mathcal{L}^c}\left( \phi  \right) = {\mathbb{E}_t}\left[ {{{\left( {{r_t} + V{_\phi}\left( {{s_{t + 1}}} \right) - V{_\phi}\left( {{s_t}} \right)} \right)}^2}} \right],
\end{equation}
where $V{_\phi}\left( {{s_t}} \right)$ denotes the value of the current state and can be estimated by the critic network.

It is worth noting that we embed the GNN layers into both the actor and critic networks to update GNN parameters, denoted as $\vartheta_a$ and $\vartheta_c$, respectively. This integration enables $\vartheta_a$ and $\vartheta_c$ to be updated concurrently with $\theta$ and $\phi$. The dimension of the output of the actor network here corresponds to the number of D2D pairs, i.e., $ \textbf{p}\in\mathbb{R}^{{M \times 1}}$. Besides, a fully connected (FC) layer is incorporated into the final layer of the critic network to produce a scalar output $ \textbf{I}\in\mathbb{R}^{{1 \times 1}}$.

The overall process of the proposed method is shown in Fig. \ref{PPO_GNN}, and the complete GNN-based RL algorithm is summarized in Algorithm 1. Although the proposed method adopts a centralized architecture, it remains applicable to certain practical wireless communication scenarios, particularly when centralized coordination and global state information are achievable. For instance, in V2I networks, roadside units or edge servers can function as centralized controllers, collecting real-time channel state information and buffer states from multiple vehicles. This centralized architecture is well-suited for V2I communication scenarios that demand low-latency and high-reliability services, such as cooperative perception and autonomous driving.

\begin{algorithm}[htb]
\caption{ RL algorithm based on GNN}\label{alg1}
\renewcommand{\algorithmicrequire}{\textbf{Input:}}
\renewcommand{\algorithmicensure}{\textbf{Output:}}
\begin{algorithmic}[1] 
\REQUIRE ~~\\ 
    Communication networks configurations, channel state information $\textbf{H}$, the number of trajectory in the memory buffer when the agent updates $N_{tra}$ times.
\ENSURE ~~\\ 
    The allocated power levels $\textbf{p}$.
    \STATE Initialize parameters: $\theta$, $\phi$.
      \FOR {training episode $i \in \left[ {1,N_t} \right]$}
            \STATE Initialize buffers of D2D pairs.
            \STATE Generate node features $\boldsymbol{\cal Y}$ and edge feature $\textbf{e}$.
            \STATE Empty the memory buffer.
      \FOR {time step $t \in \left[ {1,T_{sum}} \right]$}
            \STATE Packets arrive at buffers.
            \STATE Update buffers of D2D pairs.
            \STATE Iterate the embedding using \eqref{e13}. \\
            \STATE Obtain the allocated power levels using \eqref{adde15}. \\
            \STATE Each transmitter sends packets from its buffer.
            \STATE Update buffers of D2D pairs.
            \STATE Save the experience date to the memory buffer.
      \ENDFOR
      \IF {$i~mod~{N_{tra}} = 0$ }
            \FOR {agent update times $n \in \left[ {1,K} \right]$}
            \STATE Sample data from the memory buffer.
            \STATE Update parameters $\theta$ and $\phi$.
      \ENDFOR
      \STATE Empty the memory buffer.
      \ENDIF
      \ENDFOR
\end{algorithmic}
\end{algorithm}

In practical wireless networks, user connections frequently fluctuate, necessitating resource allocation methods with robust generalization capabilities. Unlike traditional models that require retraining as network size changes, GNNs naturally adapt to dynamic topologies through three key properties. First, weight sharing across nodes in each layer enables a trained GNN to process new nodes without parameter adjustments. Second, local neighborhood aggregation produces representations that transfer seamlessly to larger networks since the local structure remains largely consistent. Third, the inherent permutation invariance of message-passing operations ensures that model outputs depend solely on the structural relationships among nodes rather than on their ordering. By leveraging these properties, our approach scales seamlessly to varying network sizes, maintaining consistency and robustness in decision-making.

Unlike conventional graph-based methods presented in \cite{ref21, ref22, ref23, ref28}, where permutation invariance is primarily leveraged as a pre-processing step, our approach integrates permutation equivariance directly within the RL framework. By incorporating permutation equivariance with domain-specific constraints, such as the coupling between channel interference and power allocation, our method ensures that the learned policies remain robust under varying network conditions. This integration not only preserves the beneficial properties of permutation invariance but also enhances the system’s scalability to different network sizes and topologies without the need for model retraining. Mathematically, we formalize this property in the context of our integrated GNN-based RL framework as follows. Let $\boldsymbol{\Pi}  \in {\mathbb{R}^{M \times M}}$ be a permutation matrix where each row and each column contains exactly one element equal to $1$ and all other elements are $0$, and it satisfies
\begin{equation}\label{adde19}
\boldsymbol{\Pi} {\textbf{1}_M} = {\boldsymbol{\Pi} ^T}{\textbf{1}_M} = {\textbf{1}_M},
\end{equation}
where $\textbf{1}_M$ denotes a vector of size $M$ with all entries equal to one, and ${ (\cdot) ^T}$ is the transpose operation. In addition, $\mathbf{A} \in {\mathbb{R}^{M \times M}}$ denotes the weighted adjacency matrix of the graph $\mathbb{G}$, for any $\left( {i,j} \right) \in \mathbb{V}$,
\begin{equation}\label{e19}
{A_{ij}} = \left\{ {\begin{array}{*{20}{c}}
  {w\left( {i,j} \right),}&{{\text{if }}\left( {i,j} \right) \in \mathbb{V},} \\
  {0,}&{{\text{o}}{\text{.w}}{\text{.}}}
\end{array}} \right.
\end{equation}

Besides, let $\boldsymbol{\cal Y}: = {\left[ {{\cal Y}{{_1^{(0)}}^{^T}}, \cdots, {\cal Y}{{_M^{(0)}}^{^T}}} \right]^T} \in {\mathbb{R}^{M \times {F_0}}}$ denote the arrays of initial node features, ${\boldsymbol{\rm E}}: = {\left[ {{\rm E}_1^T, \cdots ,{\rm E}_M^T} \right]^T} \in {\mathbb{R}^{M \times {F_K}}}$ denote final node embedding, and $\boldsymbol{p} \in {\mathbb{R}^{M \times 1}}$ denote the allocated power levels. Thus, the following proposition is stated.
\begin{proposition}\label{proposition1}
Let $\Psi $ be a GNN operator that maps the initial node features $\boldsymbol{\cal Y} \in {\mathbb{R}^{M \times {F_0}}}$ and the weighted adjacency matrix $\mathbf{A}\in {\mathbb{R}^{M \times {M}}}$ to a set of node embeddings ${\boldsymbol{\rm E}} \in {\mathbb{R}^{M \times {F_K}}}$ through $K$ GNN lagers, i.e.,
\begin{equation}
{\boldsymbol{\rm E}} = \Psi \left( {\boldsymbol{\cal Y},{\mathbf{A}}} \right).
\end{equation}

Let $\Theta $ be an RL operator that maps the node embeddings to the allocated power levels $\boldsymbol{p} \in {\mathbb{R}^{M \times {1}}}$ via an end-to-end learning process that integrates GNN layers within the actor and critic networks, i.e.,
\begin{equation}
\boldsymbol{p} = \Theta \left( {\boldsymbol{\rm E}} \right).
\end{equation}

Then, for any given permutation matrix $\boldsymbol{\Pi}$, we have
\begin{equation}\label{e20}
\Theta  \left(  {\Psi \left( {\boldsymbol{\Pi} \boldsymbol{\cal Y},\boldsymbol{\Pi} \mathbf{A}{\boldsymbol{\Pi} ^T}} \right)} \right)= \boldsymbol{\Pi} \Theta \left({ \Psi \left( \boldsymbol{\cal Y}, \mathbf{A}\right)}\right).
\end{equation}
\end{proposition}
%
%
%
%

%
\begin{proof}[\rm{\textbf{Proof}}]
 Please see the Appendix B.
\end{proof}

\section{Simulation Results}\label{section 4}
\subsection{Wireless Network and Parameters Settings}
Based on \cite{ref26} and \cite{ref27}, we configure the parameters for the D2D communication scenario to ensure consistency with practical deployment conditions. These studies consider similar network settings, including topology, channel characteristics, and interference models, making them well-suited references for our parameter selection. Specifically, $M$ transmitters are uniformly distributed at random within a $500$ m $\times$ $500$ m square network area, ensuring a minimum distance of $75$ m between each transmitter. Each transmitter is associated with a receiver dropped uniformly at random within an annulus surrounding it, with inner and outer radii of $10$ m and $50$ m, respectively. Long-term channel fading follows a dual-slope path-loss model with a $7$ dB long-normal shadowing component. Short-term channel fading varies across time steps and is modeled using a Rayleigh distribution, employing the sum of sinusoids (SOS) technique with a pedestrian speed of $1$ m/s. The bandwidth is $10$ MHz, and the noise power spectral density is $-174$ dBm/Hz, resulting in a noise variance of $-104$ dBm. The maximum transmit power for each D2D pair is $P_{\max} = 10$ dBm. The communication duration is divided into $T_{sum}=300$ time slots, each with a length of $T=1$ ms. Unless stated otherwise, packet arrivals follow a Poisson distribution with an average arrival rate of $\lambda = 3$, namely three packets per millisecond on average, each packet having a length of $4000$ bits. For simplicity in analysis, buffer capacities are assumed to be infinite to prevent packet loss. The primary system parameters are summarized in Table \ref{parameters}. By default, all parameters are set to the values specified in Table \ref{parameters}, with the settings in each figure taking precedence whenever applicable.

\begin{table}[htbp]
\centering
\caption{Simulation parameters \cite{ref26,ref27}}\label{parameters}
\begin{tabular}{|l|l|}
\hline
\textbf{Parameter}                                 & \textbf{Value }                                 \\
\hline
Number of D2D pairs                            & ~6                                     \\
\hline
Area size                                          & 500 m $\times$ 500 m  \\
\hline
Minimum Tx-Tx distance                              & 75 m                                   \\
\hline
Tx-Rx distance                                     & 10 m $\sim$ 50 m              \\
\hline
Bandwidth                                          & 10 MHz                                 \\
\hline
Noise power spectral density                       & -174 dBm/Hz                            \\
\hline
Noise variance $N$                                 & -104 dBm                               \\
\hline
Maximum transmit power                             & 10 dBm                                \\
\hline
Communication duration                             & 300 ms                                 \\
\hline
Length of each time slot $T$                   & 1 ms                                   \\
\hline
Average arrival rate $\lambda$                 & 3 packets/ms                           \\
\hline
Packet length                                      & 4000 bits                              \\
\hline
\end{tabular}
\end{table}

As for the setting of GNN parameters, we conduct simulation experiments to investigate the impact of the number of GNN layers as shown in Fig. \ref{layer_impact} and the feature dimension, as illustrated in Table \ref{feature_dim}. As shown in Fig. \ref{layer_impact}(a), the number of GNN layers has minimal impact on the achieved average delay, as the average delay is primarily determined by network size and channel conditions. However, as shown in Fig. \ref{layer_impact}(b) the number of GNN layers significantly influences the training convergence speed, so we use $K=2$ hidden layers based on the local extremum operator in \eqref{e13}. The results in Table \ref{feature_dim} demonstrate that increasing the feature dimension from 16 to 32 significantly enhances performance by improving the model's expressive capacity. However, further increasing the dimension to 64 provides only marginal gains, and at 128, the improvement becomes negligible while incurring a notable increase in training time. Considering this trade-off between performance and computational efficiency, we set the feature dimension to 64 ($F_1=F_2=64$) to achieve a balanced trade-off between accuracy and training efficiency. Regarding the agent, both the actor and critic networks consist of three fully connected hidden layers. These layers contain 500, 250, and 120 neurons, respectively. In the setting of the PPO algorithm, the clip parameter $\epsilon$ is $0.2$, the discount rate is $0.99$, and the GAE parameter is $0.95$.

\begin{figure*}[htbp]
\begin{center}
\subfigure[The impact of the number of GNN layers on training speed.]{
\includegraphics[width=3.2in,height=2.5in]{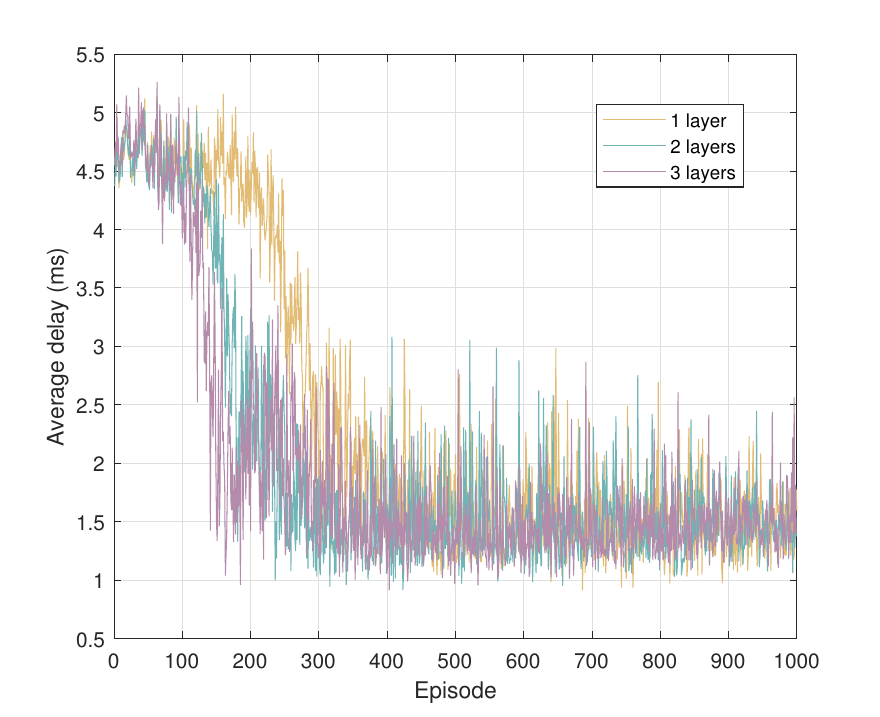}}\label{layer_impact_delay}
\subfigure[The impact of the number of GNN layers on average delay.]{
\includegraphics[width=3.2in,height=2.5in]{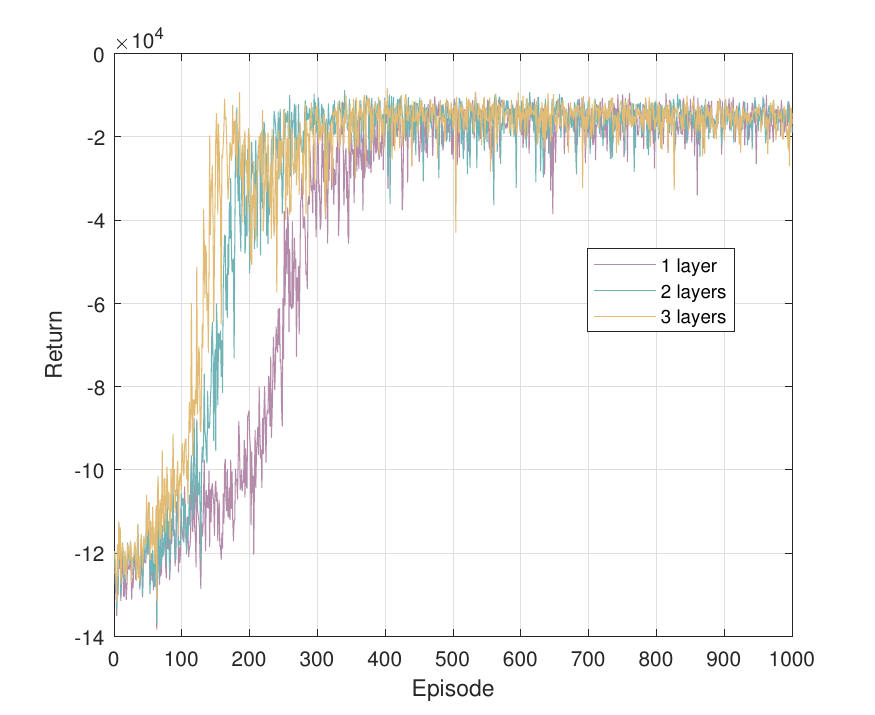}}\label{layer_impact_return}	
\caption{The impact of the number of GNN layers.} \label{layer_impact}
\end{center}
\vspace{-15pt}
\end{figure*}

\begin{table}[htbp]
\centering
\caption{Average Delay Performance with Different Feature Dimensions}\label{feature_dim}
\label{tab:feature_dim}
\begin{tabular}{c|c|c|c|c}
\hline
\textbf{Feature Dimension} & 16    & 32    & 64    & 128   \\ \hline
\textbf{Average Delay (ms)} & 1.715 & 1.439 & 1.096 & 1.078 \\ \hline
\end{tabular}
\end{table}

\subsection{Baseline Methods}
To reflect the performance of our proposed method, the following baseline methods are considered:
\begin{itemize}
\item \textbf{ITLinQ} \cite{ref34}. ITLinQ is a method designed to optimize the network topology by considering incoming interference as noise. By prioritizing high-quality links and minimizing interference, it aims to improve the overall network performance.
\item \textbf{Max power}. The transmitter in each D2D pair transmits at maximum power $P_{\max}$, and in this case, the interference between different D2D pairs is relatively large.
\item \textbf{Random power}. The transmitter in each D2D pair will randomly obtain a value within $\left[ 0, P_{\max} \right]$ as power.
\item \textbf{WMMSE} \cite{ref35}. WMMSE is an optimization algorithm used to enhance network performance by managing interference and improving signal quality. To pursue the sum rate maximization, it aims to minimize the mean square error (MSE) of the received signal while considering the weights assigned to different users.
\item \textbf{PPO} \cite{ref30}. PPO is a deep RL algorithm that balances exploration and exploitation through trust region optimization. It updates the policy by constraining the step size, ensuring stable and reliable convergence. PPO is particularly effective for continuous action spaces, making it suitable for power allocation tasks in wireless networks.
\item \textbf{PPO + GraphSAGE} \cite{addref2803}. This method integrates PPO with GraphSAGE, a GNN architecture designed for inductive learning on large-scale graphs. GraphSAGE generates low-dimensional embeddings by sampling and aggregating neighboring node features, effectively capturing local topology information.
\item \textbf{TD3} \cite{refTD3} \textbf{+ ASAP} \cite{ref29}. Twin Delayed Deep Deterministic Policy Gradient (TD3) is an actor-critic-based RL algorithm designed for continuous control tasks, known for reducing overestimation bias and improving stability.
\end{itemize}
%


\subsection{Convergence and Validity Analysis}

To validate the effectiveness and convergence of the proposed method, experiments are conducted in a scenario involving six D2D pairs. Fig. \ref{training process} shows the cumulative reward for each training episode. The figure demonstrates a progressive increase in the cumulative reward throughout the training process, indicating the effectiveness of our proposed method. Besides, the performance gradually stabilizes in around $400$ training episodes, despite some fluctuations due to channel variations and changes of packets in buffers. As shown in the figure, the proposed method consistently achieves rapid convergence regardless of the network size or traffic load. This demonstrates the scalability and effectiveness of the method under varying network conditions.
\begin{figure}[htbp]
\begin{center}
  \includegraphics[width=3.2in,height=2.5in]{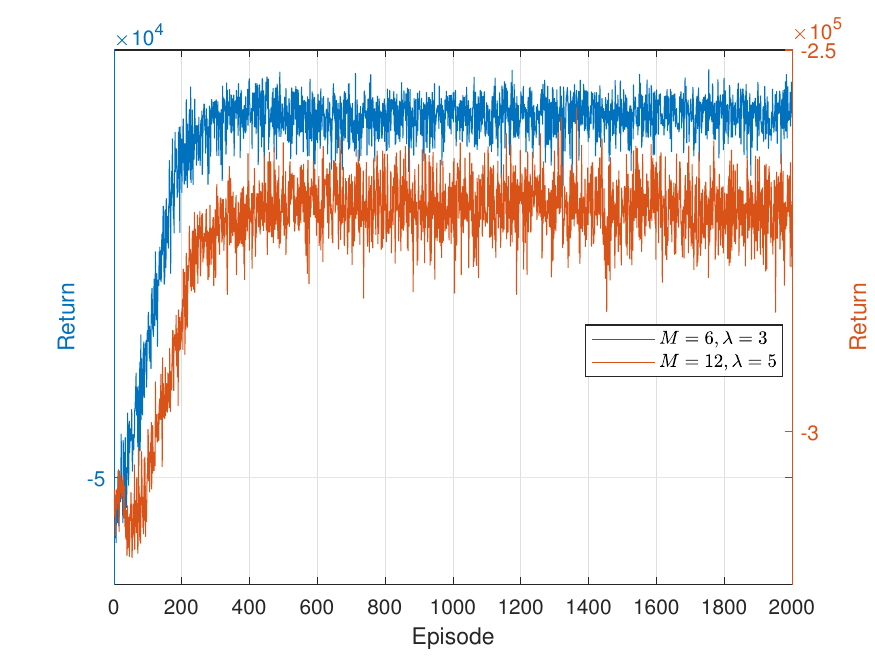}
  \caption{Return for each training episode with increasing iterations.}\label{training process}
\end{center}
\end{figure}

In Fig. \ref{performance with training}, we compare the performance of the proposed method with baseline methods in the same D2D communication scenario involving six D2D pairs from three perspectives: Average delay, 5-th percentile delay, and 95-th percentile delay. Consistent with Fig. \ref{training process}, we train the PPO network of the agent over $2000$ episodes to evaluate the performance of our method. Fig. \ref{performance with training}(a) shows the average delay of transmitted packets for six D2D pairs over $300$ time slots. It can be observed that performance gradually converges in around $400$ training episodes, as seen in Fig. \ref{training process}. The WMMSE method aims to maximize the sum rate of the entire communication network. However, this focus often leads to an imbalance in resource allocation. Users with favorable channel conditions may achieve higher transmission rates than necessary, potentially wasting power resources. Meanwhile, users with poor channel conditions are allocated insufficient power resources, leading to packet accumulation in their buffers. This imbalance negatively affects the average delay of the overall communication network. The random power method evidently fails to guarantee delay performance. Although the max power method allows all users to transmit packets at maximum power, it increases interference among users, thereby degrading performance. The ITLinQ method, which balances rate and interference, outperforms all baseline methods. In contrast, the proposed method directly aims to minimize the average delay of all D2D links and thereby learns to strategically allocate power to each link so that the transmission rates are adapted according to their queueing states. Therefore, its performance significantly surpasses baseline methods, achieving approximately a $56\%$ improvement over the ITLinQ method. Fig. \ref{performance with training}(b) and Fig. \ref{performance with training}(c) illustrate the 5-th percentile delay and the 95-th percentile delay of transmitted packets for six D2D pairs over $300$ time slots, respectively. From these two figures, it can be seen that the baseline methods, except ITLinQ, are much higher than the proposed method in terms of the 5-th percentile delay and the 95-th percentile delay. However, the proposed method sacrifices part of the 5-th percentile delay but dramatically reduces the 95-th percentile delay. This is why the proposed method outperforms baseline methods in terms of average delay, ensuring fairness among users.

In Fig. \ref{performance with training}, we conduct an ablation study to isolate and evaluate the contributions of the GNN and RL components. Specifically, we add the PPO algorithm without GNN integration, where the agent directly uses raw state features for decision-making. This setup assesses the effectiveness of the GNN in extracting network topology information. From the figure, it can be observed that the proposed method achieves faster convergence and superior performance compared to the standalone PPO algorithm. Additionally, the proposed algorithm demonstrates a noticeable performance improvement when compared with two alternative configurations: using GraphSAGE as the GNN backbone and employing TD3 as the RL algorithm. This highlights the effectiveness of the adopted GNN architecture and the PPO framework in enhancing the overall system performance.

\begin{figure*}[h]
\begin{center}
\hspace{-10mm}
\subfigure[Average delay.]{
    \includegraphics[width=2.6in,height=1.9in]{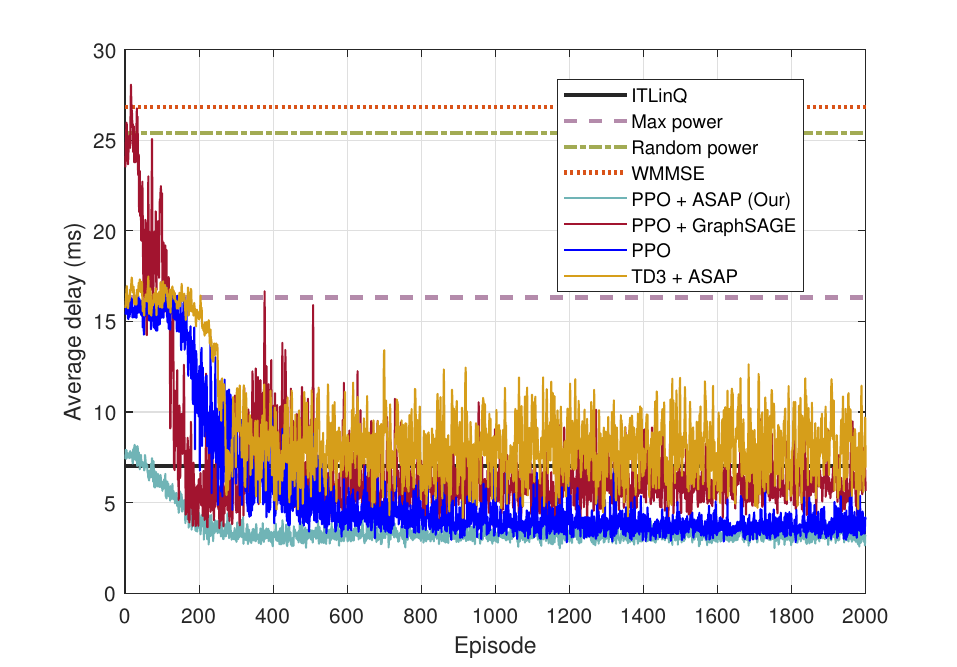}}\label{Average_delay}
\hspace{-8mm}
\subfigure[5-th percentile delay.]{
    \includegraphics[width=2.6in,height=1.9in]{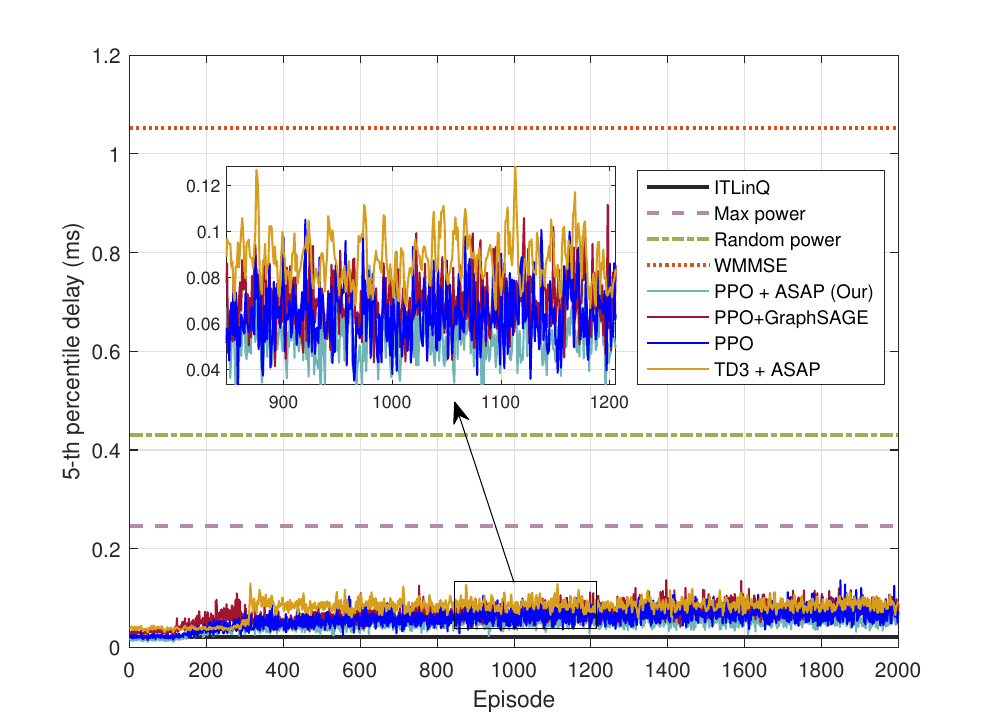}}\label{5th_percentile_delay}
\hspace{-8mm}
\subfigure[95-th percentile delay.]{
    \includegraphics[width=2.6in,height=1.9in]{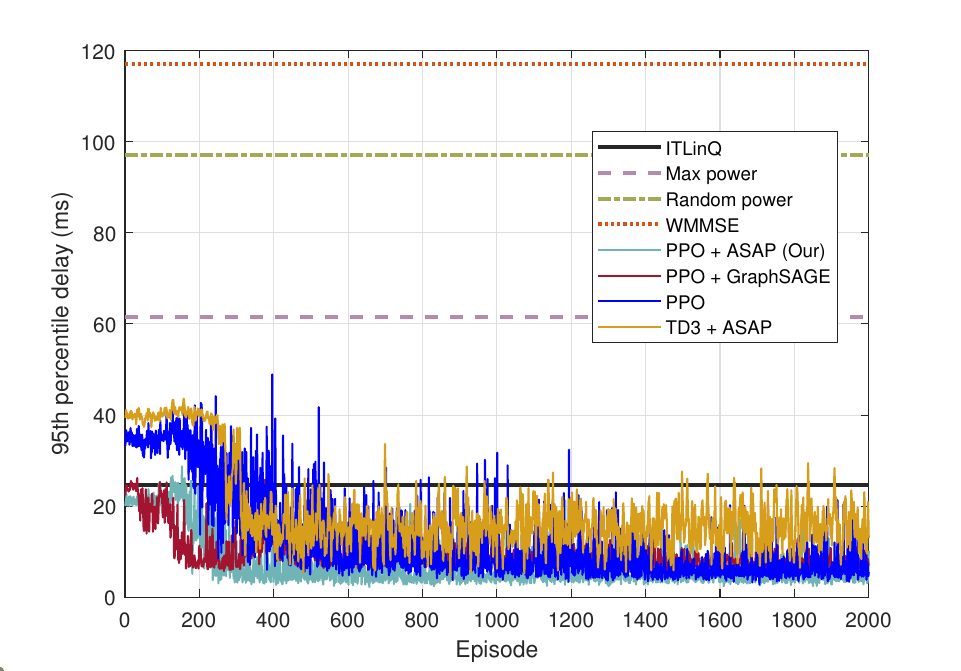}}\label{95th_percentile_delay}
\hspace{-10mm}
\caption{Performance for each training episode with increasing iterations for six D2D pairs over $300$ time slots.}\label{performance with training}
\end{center}
\vspace{-10pt}
\end{figure*}

\subsection{Performance and Fairness Evaluation}\label{Performance and Fairness Evaluation}

In Fig. \ref{packet transmission efficiency}, we compare the performance of the proposed method with baseline methods in terms of packet transmission for six D2D pairs over $300$ time slots. The numbers of transmitted packets for ITLinQ, Max power, Random power, WMMSE, our method, PPO $+$ GraphSAGE, and TD3 $+$ ASAP  are 5144, 5115, 5148, 4639, 5261, 5194, and 5218, respectively. Correspondingly, the numbers of remaining packets for these methods are 271, 320, 219, 834, 142, 228, and 294, respectively. As can be seen from the figure, the proposed method transmits the most packets and has the fewest remaining packets, demonstrating its significant superiority in packet transmission efficiency compared to the baseline methods.
\begin{figure}[htbp]
\begin{center}
  \includegraphics[width=3.2in,height=2.6in]{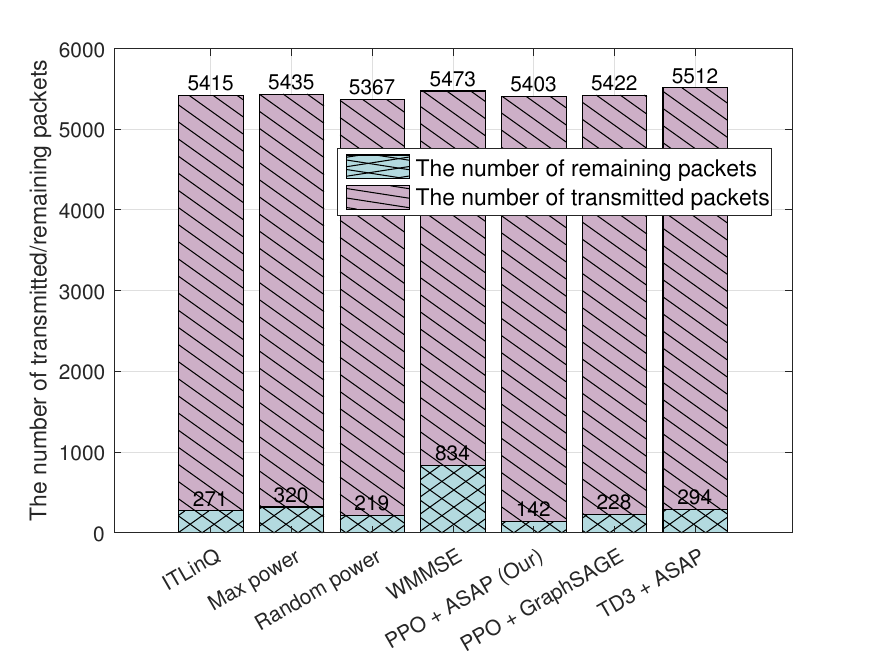}
  \caption{The number of transmitted/remaining packets for six D2D pairs over $300$ time slots.}\label{packet transmission efficiency}
\end{center}
\vspace{-10pt}
\end{figure}

We further compare the performance of our method with several baselines in terms of average delay as the number of D2D pairs varies in the communication networks, as illustrated in Fig. \ref{performance to D2D number}. The PPO network of the agent has trained over $2000$ episodes for each different number of D2D pairs and tested on the same number of D2D pairs. As the number of D2D pairs increases, interference between different users intensifies, affecting communication rates and resulting in longer waiting times of packets in buffers, as shown in Fig. \ref{performance to D2D number}. The performance trends of different baseline methods align with those in Fig. \ref{performance with training}. In addition, we observe that as the number of users in the network increases, interference grows significantly, leading to a noticeable degradation in the performance of baseline methods and a sharp increase in average delay. When the number of D2D pairs increases from $4$ to $12$, the performance of WMMSE declines by approximately $50.85\%$, Random power by approximately $46.42\%$, Max power by approximately $44.49\%$, and ITLinQ by approximately $45.98\%$. While our method also experiences performance degradation, it still outperforms all baseline methods.
\begin{figure}[htbp]
\begin{center}
  \includegraphics[width=3.2in,height=2.4in]{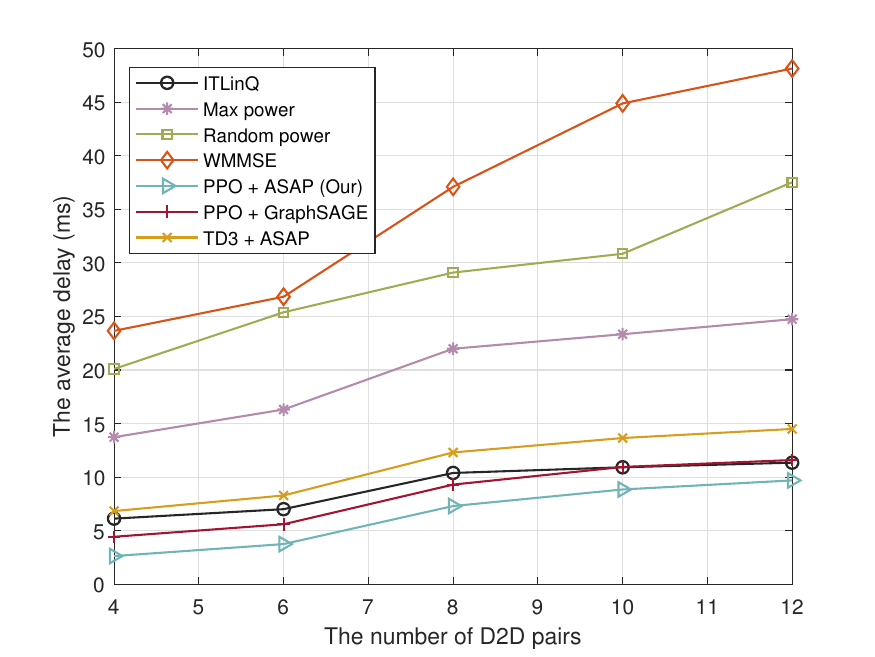}
  \caption{Average delay performance as the number of D2D pairs $M$ varies.}\label{performance to D2D number}
\end{center}
\end{figure}

Fig. \ref{arrival rate} shows the performance comparison between the proposed method and baseline methods when the arrival rate $\lambda$ changes, with the number of D2D pairs $M=6$. From the figure, it can be observed that as the arrival rate increases, i.e., the number of packets arriving at buffers per time slot rises, the accumulation of packets in each buffer intensifies, leading to a gradual increase in average delay. However, since the proposed method considers the number of accumulated packets in buffers as part of the state of the agent, it can effectively mitigate the average delay.
\begin{figure}[htbp]
\begin{center}
  \includegraphics[width=3.2in,height=2.5in]{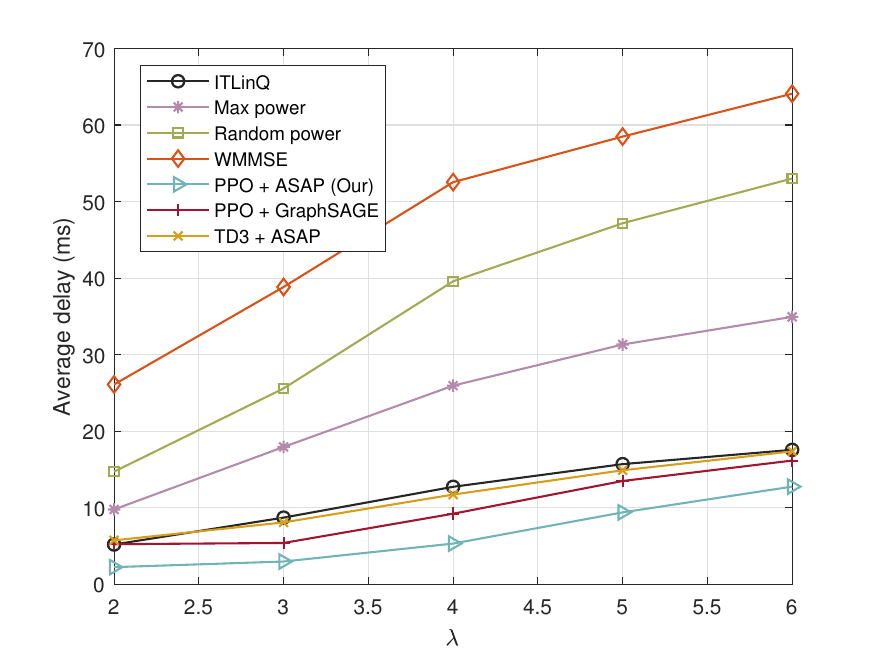}
  \caption{Average delay performance as arrival rate $\lambda$ varies.}\label{arrival rate}
\end{center}
\end{figure}

To illustrate why the proposed method achieves better performance compared with baseline methods like WMMSE, we select an episode and present the remaining packets in each buffer for every D2D pair after each time step in Fig. \ref{Remaining packets}. To provide clearer results, we focus on a communication scenario with four D2D pairs. Fig. \ref{Remaining packets}(a) shows the remaining packets in each buffer after each time step for the proposed method. Overall, the number of remaining packets in each buffer is relatively small. For instance, D2D pair 3, despite having poorer performance, only has around 80 remaining packets at most. Additionally, from the figure, it can be observed that the number of remaining packets fluctuates frequently at each time step, indicating that almost all D2D pairs are actively transmitting packets. This demonstrates the fairness of the proposed method, with detailed insights to be further elaborated in Fig. \ref{D2D transmission rates}. It shows the remaining packets in each buffer after each time step for WMMSE in Fig. \ref{Remaining packets}(b). Generally, there are relatively more remaining packets in each buffer. Using D2D pair 3 as an example again, the number of remaining packets can reach as high as $600$. Moreover, the figure shows that the remaining number of packets for each D2D pair increases over a period, indicating that during that time, they have transmitted few or even no packets, resulting in an accumulation of packets in the buffer.

\begin{figure*}[htbp]
\begin{center}
\subfigure[The remaining packets in each buffer of the proposed method.]{
\includegraphics[width=3.2in,height=2.5in]{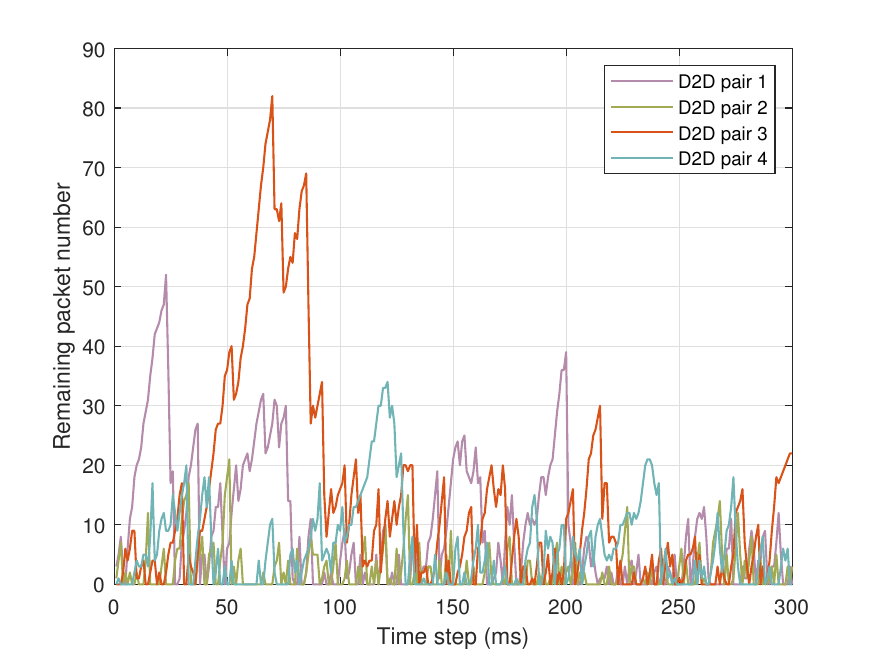}
}
\subfigure[The remaining packets in each buffer of WMMSE.]{
\includegraphics[width=3.2in,height=2.5in]{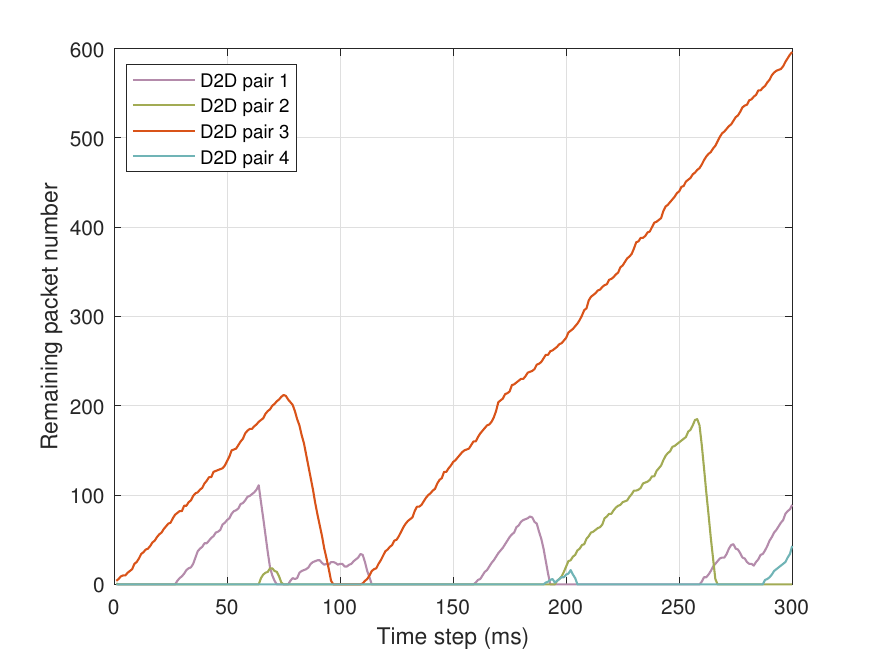}
}	
\caption{The change of the remaining packets in each buffer over time step of the proposed method and WMMSE.}
\label{Remaining packets}

\subfigure[D2D transmission rates of the proposed method.]{
\includegraphics[width=3.2in,height=2.5in]{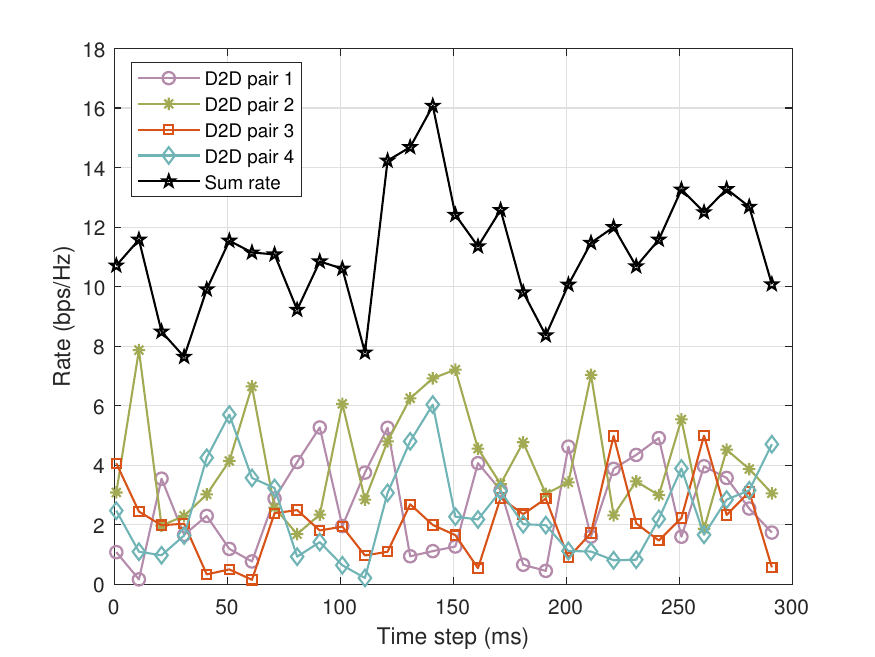}
}
\subfigure[D2D transmission rates of WMMSE.]{
\includegraphics[width=3.2in,height=2.5in]{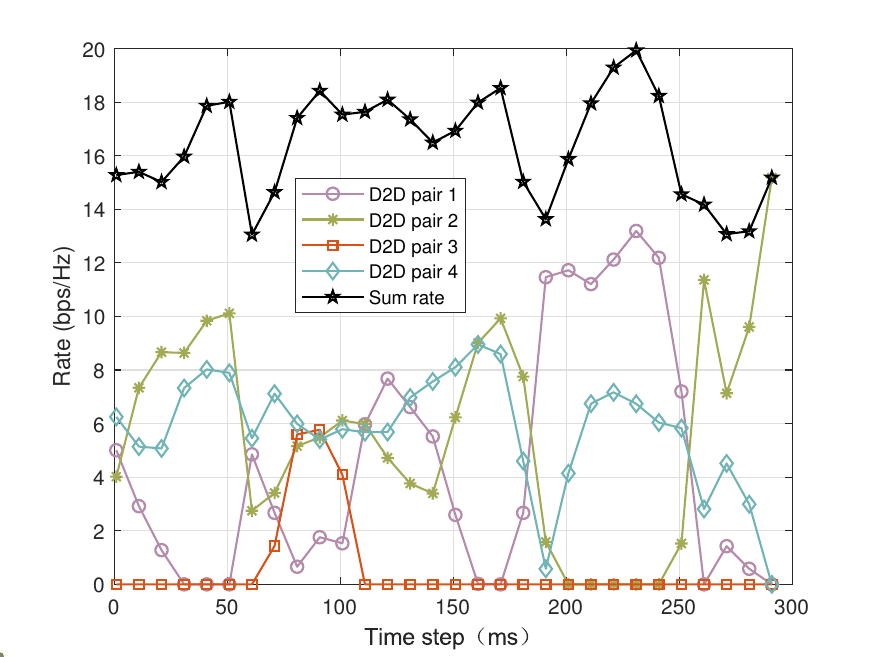}
}	
\caption{D2D transmission rates of the proposed method and WMMSE within the same episode as Fig. \ref{Remaining packets}.}
\label{D2D transmission rates}
\end{center}
\vspace{-15pt}
\end{figure*}
\begin{table*}[h]
\centering
\caption{Scalability of our method on average delay for large network size but similar network density.}\label{similar_density}
\begin{tabular}{|c|c|c|c|c|c|c|c|c|}
\hline
$M$  & Area ($m^2$) & Our method (ms) & PPO $+$ GraphSAGE & TD3 $+$ ASAP  & ITLinQ & Max power& Random power & WMMSE   \\
\hline
6  & 500     & 5.374 (trained)  & 6.609 (trained) & 6.994 (trained) & 7.078  & 12.664    & 21.542       & 34.412  \\ \cline{1-2}
12 & 700     & 5.565           &     6.909      &   7.308    & 7.439  & 14.182    & 20.351       & 37.686  \\ \cline{1-2}
24 & 1000   & 5.673           &     7.294      &   8.854   & 8.164  & 18.458    & 27.849       & 46.268  \\ \cline{1-2}
54 & 1500   & 6.319           &     8.556      &   9.701    & 9.625   & 19.768    & 29.865       & 51.769  \\ \cline{1-2}
\hline
\end{tabular}
\end{table*}
\begin{table*}[htbp]
\centering
\caption{Scalability of our method on average delay for large network size and higher network density.}\label{higher_density}
\begin{tabular}{|c|c|c|c|c|c|c|c|c|}
\hline
$M$  & Area ($m^2$)        & Our method (ms) & PPO $+$ GraphSAGE & TD3 $+$ ASAP  & ITLinQ & Max power & Random power & WMMSE  \\ \hline
6  & \multirow{5}{*}{500}  & 5.374 (trained) &6.609 (trained)   & 6.994 (trained) & 7.078       & 14.660    & 21.546       & 34.412 \\ \cline{1-1} 
8  &                      & 7.424        &    7.926    &   8.380        & 8.406       & 17.008         & 25.130         & 43.334 \\ \cline{1-1} 
10 &                      & 8.707        &    9.049    &   9.225          & 9.247       & 19.817         & 26.827        & 47.693 \\ \cline{1-1} 
12 &                      & 13.238      &     14.259    &  14.864          & 14.352      & 29.096         & 43.594            & 62.576 \\ \hline
\end{tabular}
\end{table*}

To further demonstrate the fairness guaranteed by the proposed method, we compare in Fig. \ref{D2D transmission rates} the instantaneous rates of all D2D pairs at each time step between our method and the WMMSE in the same episode as Fig. 7. We averaged the instantaneous rates over every $10$ time steps to provide a clear presentation for both methods. From Fig. \ref{D2D transmission rates}(a), it can be observed that the rates of all users fluctuate frequently, and overall, there is not a significant difference in rates among each user, even for D2D pair $3$ with poor channel conditions. In contrast, Fig. \ref{D2D transmission rates}(b) shows that almost all users experience periods of zero rate, especially D2D pair $3$, where there are extended intervals with zero rates, resulting in packet accumulation. While WMMSE does achieve a higher sum rate, the difference is not substantial compared to the proposed method. Fig. \ref{Remaining packets} and Fig. \ref{D2D transmission rates} clearly demonstrate that the proposed method effectively reduces average delay and ensures fairness through power control strategies.

\subsection{Scalability Evaluation with Varying Network Densities}

We evaluate the scalability of our method in unseen networks with a larger number of D2D pairs while maintaining a consistent network density, as shown in Table \ref{similar_density}. This is achieved by proportionally increasing the area size. Specifically, we adjust the edge length of the square area to $\sqrt {M/24} \times 1$ km, ensuring a constant density of 24 users/km$^2$, which is consistent with the scenario where six D2D pairs are distributed over a 500 m $\times$ 500 m square area. The agent is trained on networks with six D2D pairs over a 500 m $\times$ 500 m square area and then tested on networks with different numbers of D2D pairs and area sizes. We change the random seed during testing in the six D2D pair to ensure the agent has not encountered the exact scenario before. The performance of the baseline methods is obtained directly by applying them to the current network settings. As detailed in Table \ref{similar_density}, the performance of our method outperforms baseline methods on average delay and remains stable in unseen networks with a larger number of D2D pairs while maintaining similar network density.

In Table \ref{higher_density}, we assess the scalability of our method in unseen networks with higher network density. The experiment maintains a constant square area size while increasing the number of D2D pairs. The agent is initially trained on networks with six D2D pairs and tested on denser networks within the same 500 m $\times$ 500 m area. The performance of the baseline methods is directly obtained by applying them to the current network settings. The table shows that the proposed method demonstrates robust scalability, especially as network density increases. As network density rises, the interference among different D2D pairs intensifies, affecting the instantaneous rate and increasing the average delay. Nevertheless, the performance of our method remains acceptable and consistently outperforms the baseline methods in terms of average delay.

\subsection{Generalization Across Varying User Distributions}

We evaluate the generalization of our method with respect to different user distributions in Table \ref{Generalization}. The agent is trained on networks with six D2D pairs located within an annular region with a radius ranging from 10 m to 50 m over a square area of 500 m $\times$ 500 m. We then test the agent on scenarios where the annular region radius varies while maintaining six D2D pairs within the same 500 m $\times$ 500 m area. As illustrated in Table \ref{Generalization}, the proposed method performs effectively across a range of radii, including 2 m $\sim$ 65 m, 30 m $\sim$ 70 m, and all 30 m, without further training. This demonstrates the adaptability of GNN parameterizations in our method, enabling it to be applicable to different scenarios. It differs from the approaches relying solely on RL, which requires retraining the agent when deployed in the new network configurations.

\begin{table}
\centering
\caption{Generalization of our method on average delay for different user distributions.}\label{Generalization}
\begin{tabular}{|c|c|c|c|c|}
\hline
D2D distance (m)   & 2$\sim$65  &  10$\sim$50 (trained)   &   30$\sim$70   & all 30  \\
\hline
Average delay (ms) & 5.415      &  5.374                  &   5.791        & 4.755   \\
\hline
\end{tabular}
\end{table}

\section{Conclusion and future research}\label{section 5}
This paper has addressed the challenge of balancing user fairness in D2D communication by proposing a novel power allocation approach based on GNNs and RL. Our approach extended the traditional state information to include comprehensive factors such as channel state information, packet delay, the number of backlogged packets, and the number of transmitted packets, providing a more detailed representation of the network topologies and states. We adopted a centralized RL method with a central controller acting as an agent, trained using the PPO algorithm. By embedding GNN layers into both the actor and critic networks of the PPO algorithm, we enhanced the scalability and generalization of the proposed method. This integration enabled efficient parameter updates of GNNs and the generation of a low-dimensional embedding from the state information, which was crucial for optimizing power allocation strategies. Our extensive numerical evaluations demonstrated that the proposed method significantly reduced average delay while ensuring user fairness, outperforming baseline methods. The superior scalability and generalization capability of our approach highlighted its potential for practical deployment in various D2D communication network configurations. This work provides a robust solution to the delay-power trade-off challenge, leveraging advanced machine learning techniques to optimize resource allocation in next-generation wireless networks.

Future research could explore distributed and partially observable RL methods by adopting distributed GNN-based RL architectures, where each node (or local cluster) performs localized inference and exchanges partial information with its neighbors. Instead of transmitting raw state information, each node locally processes and transmits parameterized information, reducing the data transmission load while preserving essential system state representations. This reduces the reliance on a central controller while leveraging the message-passing capability of GNNs to efficiently disseminate local state information.

\begin{appendices}
\section{\footnotesize{PROOF OF PROPOSITION 1}}
According to the reward function in \eqref{e16}, for each episode that covers the communication duration during training, the cumulative reward $C$ can be expressed as
\begin{equation}\label{ade25}
\begin{aligned}
C &= \sum\limits_{n = 1}^{{T_{sum}}} {{R_n}}  =  - \sum\limits_{n=1}^{{T_{sum}}} {\sum\limits_{i=1}^M {{J_i}\left[ n \right]} }  \\
&=  - \sum\limits_{i=1}^M {\left( {\sum\limits_{n=1}^{{T_{sum}}} {{J_i}\left[ n \right]} } \right)} \mathop  \approx \limits^{\Circled{1}}  - \sum\limits_{i=1}^M {\left( {\sum\limits_{g = 1}^{m\left( i \right)} {{T_{{b_i}}}\left( g \right)} } \right)}  \\
&\mathop  \approx \limits^{\Circled{2}}  - \sum\limits_{i=1}^M {\left( {\sum\limits_{g = 1}^{m\left( i \right)} {{T_{{D_i}}}\left( g \right)} } \right)} , \\
\end{aligned}
\end{equation}
where {\footnotesize{\Circled{1}}} follows from the fact that if a packet is not transmitted immediately, it needs to wait in the buffer, with its waiting time increasing by one for each additional time slot, and the value that the number of packets in each time slot is then counted and summed across all time slots, corresponds to the cumulative reward. {\footnotesize{\Circled{2}}} holds since if the allocated transmit power is sufficient for packet transmission, the transmission delay equals the packet length divided by the transmission rate, which must be less than 1 ms. According to the simulation results, the minimum average delay is approximately 5 ms, meaning that the impact of queueing delay on communication performance is significantly greater than that of transmission delay.

Through multiple training episodes, we expect to maximize returns, i.e., to maximize cumulative reward, thereby
\begin{equation}\label{ade26}
\begin{aligned}
\max C & \approx \max  - \sum\limits_{i = 1}^M {\sum\limits_{g = 1}^{m\left( i \right)} {{T_{{D_i}}}\left( g \right)} }  = \min \sum\limits_{i = 1}^M {\sum\limits_{g = 1}^{m\left( i \right)} {{T_{{D_i}}}\left( g \right)} }  \\
& \mathop  \Leftrightarrow \limits^{\Circled{3}} \min \underbrace {\frac{1}{M}\sum\limits_{i = 1}^M {\frac{1}{{\bar m\left( i \right)}}} \sum\limits_{g = 1}^{m\left( i \right)} {{T_{{D_i}}}\left( g \right)} }_{\Circled{4}}, \\
\end{aligned}
\end{equation}
where {\footnotesize{\Circled{3}}} holds since that given the average arrival rate $\lambda$, and the communication duration of $T_{sum}$ time slots, the average number of packets $\bar m\left( i \right)$ arriving at the buffer $i$ during remains constant, and the number of D2D pairs $M$ is a constant. Furthermore, considering the truncation effect, there are still packets left in the buffer at the end of the entire communication duration, the average delay we actually calculate in \eqref{e16} is larger than the value of {\footnotesize{\Circled{4}}}. However, due to the minimization operation, the objective function and cumulative reward, are brought as close as possible. Hence, the proof is completed. $\hfill\blacksquare$

\section{\footnotesize{PROOF OF PROPOSITION 2}}
Since the architecture of message-passing GNNs is constructed by stacking layers composed of multiple aggregation and combination operations, as shown in (\ref{e13}), it suffices to show that the GNN operator $\Psi_l $ at any layer $l \in \left\{ {1, \cdots, L} \right\}$ is permutation-equivariant. For the layer $l-1$ in the unpermuted graph, let $\boldsymbol{\cal Y}^{(l-1)}: = {\left[ {{\cal Y}{{_1^{(l-1)}}^{^T}}, \cdots, {\cal Y}{{_M^{(l-1)}}^{^T}}} \right]^T} \in {\mathbb{R}^{M \times {F_{l-1}}}}$ denotes the array of node features. Moreover, for the layer $l-1$ in the permuted graph, let $\boldsymbol{\cal \tilde Y}^{(l-1)}: = {\left[ {{\cal \tilde Y}{{_1^{(l-1)}}^{^T}}, \cdots, {\cal \tilde Y}{{_M^{(l-1)}}^{^T}}} \right]^T} = {\left[ {{\cal Y}{{_{\pi (1)}^{(l-1)}}^{^T}}, \cdots, {\cal Y}{{_{\pi (M)}^{(l-1)}}^{^T}}} \right]^T} = \boldsymbol{\Pi} \boldsymbol{\cal Y}^{(l-1)} $ denotes the array of node features. $\pi (\cdot)$ denotes the permutation operator corresponding to the permutation matrix $\boldsymbol{\Pi}$, which swapping the positions of nodes $v \in \mathbb{V} $ forms a new set of nodes $\tilde v \in \mathbb{\tilde V}$. Let ${\mathcal{N}_{\tilde v}^{(1)}}$ denotes 1-hop neighbors of any node $\tilde v \in \mathbb{\tilde V}$ in the permuted graph, and for any edge $(\tilde u, \tilde v)$, let $\tilde e_{\tilde u, \tilde v}$ denotes its weight. For simplicity, ignoring the parameter vector $\vartheta$, the feature vector of any node $\tilde v \in \mathbb{\tilde V}$ at layer $l$ in the permuted graph can be written as
\begin{equation}
\begin{small}
\begin{array}{l}
\tilde {\cal Y}_{\tilde v}^{\left( l \right)} = {\Psi ^{\left( l \right)}}\left( {\tilde {\cal Y}_{\tilde v}^{\left( {l - 1} \right)},{e_{\tilde v,\tilde v}},{{\{ \tilde {\cal Y}_{\tilde u}^{\left( {l - 1} \right)},{e_{\tilde u,\tilde v}}\} }_{\tilde u \in {{ {\cal N}}_{\tilde v}^{(1)}}}}} \right)\\
\mathop  = \limits^{\Circled{5}} {\Psi ^{\left( l \right)}}\left( {{\cal Y}_{\pi \left( v \right)}^{\left( {l - 1} \right)},{e_{\pi \left( v \right),\pi \left( v \right)}},{{\{ {\cal Y}_{\pi \left( u \right)}^{\left( {l - 1} \right)},{e_{\pi \left( u \right),\pi \left( v \right)}}\} }_{u \in \mathbb{\tilde V} :\pi \left( u \right) \in {{\cal N}_{\pi \left( v \right)}^{(1)}}}}} \right)\\
\mathop  = \limits^{\Circled{6}} {\Psi ^{\left( l \right)}}\left( {{\cal Y}_{\pi \left( v \right)}^{\left( {l - 1} \right)},{e_{\pi \left( v \right),\pi \left( v \right)}},\left\{ {{\cal Y}_u^{\left( {l - 1} \right)},{e_{u,\pi \left( v \right)}}{\} _{u \in {{\cal N}_{\pi \left( v \right)}^{(1)}}}}} \right\}} \right)\\
\mathop  = \limits^{\Circled{7}} {\cal Y}_{\pi \left( v \right)}^{\left( l \right)},
\end{array}
\end{small}
\end{equation}
where {\footnotesize{\Circled{5}}} is true because ${{\cal N}_{\tilde v}^{(1)}} = \left\{ {\tilde u \in \mathbb{\tilde V} :\pi \left( {\tilde u} \right) \in {{\cal N}_{\pi \left( {\tilde v} \right)}^{(1)}}} \right\}$ for any node $\tilde v \in \mathbb{\tilde V}$  and ${e_{\tilde u,\tilde v}} = {e_{\pi \left( u \right),\pi \left( v \right)}}$ for any edge $(\tilde u,\tilde v)$ in the permuted graph. Additionally, {\footnotesize{\Circled{6}}} holds due to a change of variables $u \leftarrow \pi \left( u \right)$, and {\footnotesize{\Circled{7}}} follows from the definition of the aggregation and combination operators in (\ref{e13}) for node $\pi (v)$. This implies that $\tilde {\cal Y}^{(l)}=\Pi {\cal Y}^{(l)}$, i.e., the output feature vectors given the permuted inputs is the accordingly-permuted version of the output feature vectors in the unpermuted graph, hence completing the proof. This means that $\tilde {\cal Y}^{(l)}=\Pi {\cal Y}^{(l)}$, i.e., the output feature vector for the permuted input is the corresponding permuted version of the output feature vector from the unpermuted graph, thereby completing the proof. $\hfill\blacksquare$

\end{appendices}

\newpage

\vfill


\begin{thebibliography}{}
\providecommand{\url}[1]{#1}
\csname url@samestyle\endcsname
\providecommand{\newblock}{\relax}
\providecommand{\bibinfo}[2]{#2}
\providecommand{\BIBentrySTDinterwordspacing}{\spaceskip=0pt\relax}
\providecommand{\BIBentryALTinterwordstretchfactor}{4}
\providecommand{\BIBentryALTinterwordspacing}{\spaceskip=\fontdimen2\font plus
\BIBentryALTinterwordstretchfactor\fontdimen3\font minus
  \fontdimen4\font\relax}
\providecommand{\BIBforeignlanguage}[2]{{%
\expandafter\ifx\csname l@#1\endcsname\relax
\typeout{** WARNING: IEEEtran.bst: No hyphenation pattern has been}%
\typeout{** loaded for the language `#1'. Using the pattern for}%
\typeout{** the default language instead.}%
\else
\language=\csname l@#1\endcsname
\fi
#2}}
\providecommand{\BIBdecl}{\relax}
\BIBdecl

\end{thebibliography}


\begin{thebibliography}{1}
\bibliographystyle{IEEEtran}

\bibitem{ref1} H. ElSawy, E. Hossain, and M. -S. Alouini, ``Analytical modeling of mode selection and power control for underlay D2D communication in cellular networks," \emph{IEEE Trans. Commun.}, vol. 62, no. 11, pp. 4147-4161, Nov. 2014.

\bibitem{ref2} M. Agiwal, A. Roy, and N. Saxena, ``Next generation 5G wireless networks: A comprehensive survey," \emph{IEEE Commun. Surveys Tuts.}, vol. 18, no. 3, pp. 1617-1655, Feb. 2016.

\bibitem{ref3} IMT-2020 (5G) Promotion Group, ``White paper on 5G bearer requirements," CAICT, Jun. 2018.

\bibitem{ref4} C. -X. Wang, X. You, X. Gao, et al., ``On the road to 6G: Visions, requirements, key technologies, and testbeds," \emph{IEEE Commun. Surveys Tuts.}, vol. 25, no. 2, pp. 905-974, Feb. 2023.

\bibitem{ref5} S. Parkvall, E. Dahlman, A. Furuskar, and M. Frenne, ``NR: The new 5G radio access technology," \emph{IEEE Commun. Stand. Mag.}, vol. 1, no. 4, pp. 24-30, Dec. 2017.

\bibitem{ref6} O. Liberg, C. Hoymann, C. Tidestav, D. C. Larsson, I. Rahman, R. Blasco, S. Falahati, and Y. Blankenship, ``Introducing 5G advanced," \emph{IEEE Commun. Stand. Mag.}, vol. 8, no. 1, pp. 52-57, Mar. 2024.

\bibitem{ref7} Y. Zhang, W. Cheng, and W. Zhang, ``Multiple access integrated adaptive finite blocklength for ultra-low delay in 6G wireless networks," \emph{IEEE Trans. Wireless Commun.}, vol. 23, no. 3, pp. 1670-1683, Mar. 2024.

\bibitem{ref8} C. Guo, L. Liang and G. Y. Li, ``Resource allocation for high-reliability low-latency vehicular communications with packet retransmission," \emph{IEEE Trans. Veh. Technol.}, vol. 68, no. 7, pp. 6219-6230, Jul. 2019.


\bibitem{ref9} M. Waqas, Y. Niu, Y. Li, M. Ahmed, D. Jin, S. Chen, and Z. Han, ``A comprehensive survey on mobility-aware D2D communications: Principles, practice and challenges," \emph{IEEE Commun. Surveys Tuts.}, vol. 22, no. 3, pp. 1863-1886, Jun. 2019.

\bibitem{ref10} J. Liu, W. Chen, and K. B. Letaief, ``Delay optimal scheduling for ARQ-aided power-constrained packet transmission over multi-state fading channels," \emph{IEEE Trans. Wireless Commun.}, vol. 16, no. 11, pp. 7123-7137, Nov. 2017.

\bibitem{ref11} W. Chen, Z. Cao, and K. B. Letaief, ``Optimal delay-power tradeoff in wireless transmission with fixed modulation," in \emph{Proc. IWCLD}, 2007, pp. 60-64.

\bibitem{addref11} L. Liang, G. Y. Li, and W. Xu, ``Resource allocation for D2D-enabled vehicular communications,'' \emph{IEEE Trans. Commun.}, vol. 65, no. 7, pp. 3186-3197, Jul. 2017.

\bibitem{ref12} W. Bao, H. Chen, Y. Li, and B. Vucetic, ``Joint rate control and power allocation for non-orthogonal multiple access systems," \emph{IEEE J. Sel. Areas Commun.}, vol. 35, no. 12, pp. 2798-2811, Dec. 2017.

\bibitem{ref13} O. L. Alcaraz López, H. Alves, and M. Latva-aho, ``Joint power control and rate allocation enabling ultra-reliability and energy efficiency in SIMO wireless networks," \emph{IEEE Trans. Commun.}, vol. 67, no. 8, pp. 5768-5782, Aug. 2019.

\bibitem{ref14} C. Guo, L. Liang, and G. Y. Li, ``Resource allocation for low-latency vehicular communications: An effective capacity perspective," \emph{IEEE J. Sel. Areas Commun.}, vol. 37, no. 4, pp. 905-917, Apr. 2019.

\bibitem{ref15} H. Ye, G. Y. Li, and B. -H. F. Juang, ``Deep reinforcement learning based resource allocation for V2V communications," \emph{IEEE Trans. Veh. Technol.}, vol. 68, no. 4, pp. 3163-3173, Apr. 2019.

\bibitem{addref15} L. Liang, H. Ye, and G. Y. Li, ``Spectrum sharing in vehicular networks based on multi-agent reinforcement learning," \emph{IEEE J. Sel. Areas Commun.}, vol. 37, no. 10, pp. 2282-2292, Oct. 2019.

\bibitem{ref16} X. Deng, J. Yin, P. Guan, N. N. Xiong, L. Zhang, and S. Mumtaz, ``Intelligent delay-aware partial computing task offloading for multiuser industrial internet of things through edge computing," \emph{IEEE Internet Things J.}, vol. 10, no. 4, pp. 2954-2966, Feb. 2023.

\bibitem{ref17} Y. Zhao, Y. Kim, and J. Lee, ``SOQ: Structural reinforcement learning for constrained delay minimization with channel state information," \emph{IEEE Internet Things J.}, vol. 11, no. 3, pp. 4628-4644, Feb. 2024.

\bibitem{ref18} T. N. Kipf and M. Welling, ``Semi-supervised classification with graph convolutional networks," in \emph{Proc. ICLR}, 2017.

\bibitem{ref19} P. Veličković, G. Cucurull, A. Casanova, A. Romero, P. Liò, and Y. Bengio, ``Graph attention networks," in \emph{Proc. ICLR}, 2018.

\bibitem{ref20} J. Zhou, G. Cui, S. Hu, Z. Zhang, C. Yang, Z. Liu, L. Wang, C. Li, and M. Sun, ``Graph neural networks: A review of methods and applications," \emph{AI Open}, vol. 1, pp. 57-81, 2020.

\bibitem{ref21} Z. Wang, M. Eisen, and A. Ribeiro, ``Learning decentralized wireless resource allocations with graph neural networks," \emph{IEEE Trans. Signal Process.}, vol. 70, pp. 1850-1863, 2022.

\bibitem{ref22} M. Eisen and A. Ribeiro, ``Optimal wireless resource allocation with random edge graph neural networks," \emph{IEEE Trans. Signal Process.}, vol. 68, pp. 2977-2991, 2020.

\bibitem{ref23} Y. Shen, Y. Shi, J. Zhang, and K. B. Letaief, ``Graph neural networks for scalable radio resource management: Architecture design and theoretical analysis," \emph{IEEE J. Sel. Areas Commun.}, vol. 39, no. 1, pp. 101-115, Jan. 2021.

\bibitem{ref24} M. Lee, G. Yu, and H. Dai, ``Decentralized inference with graph neural networks in wireless communication systems," \emph{IEEE Trans. Mobile Comput.}, vol. 22, no. 5, pp. 2582-2598, May 2023.

\bibitem{ref25} Y. Shen, J. Zhang, S. H. Song, and K. B. Letaief, ``Graph neural networks for wireless communications: From theory to practice," \emph{IEEE Trans. Wireless Commun.}, vol. 22, no. 5, pp. 3554-3569, May 2023.

\bibitem{ref26} N. NaderiAlizadeh, M. Eisen and A. Ribeiro, ``Learning resilient radio resource management policies with graph neural networks," \emph{IEEE Trans. Signal Process.}, vol. 71, pp. 995-1009, 2023.

\bibitem{add1} D. E. Rumelhart, G. E. Hinton, and R. J. Williams, ``Learning representations by back-propagating errors," \emph{Nature}, vol. 323, no. 6088, pp. 533-536, 1986.

\bibitem{add2} S. Hochreiter and J. Schmidhuber, ``Long short-term memory," \emph{Neural Comput.}, vol. 9, no. 8, pp. 1735-1780, 1997.

\bibitem{add3} A. Vaswani et al., ``Attention is all you need," in \emph{Proc. NeurIPS}, vol. 30, 2017, pp. 5998--6008.

\bibitem{add4}  M. Mohsenivatani, S. Ali, V. Ranasinghe, N. Rajatheva and M. Latva-Aho, ``Graph representation learning for wireless communications," \emph{IEEE Commun. Mag.}, vol. 62, no. 1, pp. 141-147, Jan. 2024.

\bibitem{add5} Y. Lu, Y. Li, R. Zhang, W. Chen, B. Ai and D. Niyato, ``Graph neural networks for wireless networks: Graph representation, architecture and evaluation," \emph{IEEE Wireless Commun.}, vol. 32, no. 1, pp. 150-156, Feb. 2025.

\bibitem{ref27} N. NaderiAlizadeh, M. Eisen, and A. Ribeiro, ``State-augmented learnable algorithms for resource management in wireless networks," \emph{IEEE Trans. Signal Process.}, vol. 70, pp. 5898-5912, 2022.

\bibitem{ref28} V. Lima, M.Eisen, K. Gatsis, and A. Ribeiro, ``Large-scale graph reinforcement learning in wireless control systems," \emph{arXiv preprint, arXiv:2201.09859v2}, 2022.

\bibitem{addref2801} F. Yang, C. Yang, J. Huang, et al., ``Mutual-interference-aware throughput enhancement in massive IoT: A graph reinforcement learning framework," \emph{IEEE Internet Things J.}, vol. 11, no. 18, pp. 30341-30353, Sep. 2024.

\bibitem{addref2802} Z. Sun, Y. Mo and C. Yu, ``Graph-reinforcement-learning-based task offloading for multiaccess edge computing," \emph{IEEE Internet of Things J.}, vol. 10, no. 4, pp. 3138-3150, Feb. 2023.

\bibitem{addref28} R. Paul, K. Cohen, and G. Kedar, ``Multi-flow transmission in wireless interference networks: A convergent graph learning approach," \emph{IEEE Trans. Wireless Commun.}, vol. 23, no. 4, pp. 3691-3705, Apr. 2024.

\bibitem{addref2803} Z. Liu, J. Zhang, E. Shi, et al., ``Graph neural network meets multi-agent reinforcement learning: Fundamentals, applications, and future directions," \emph{IEEE Wireless Commun.}, vol. 31, no. 6, pp. 39-47, Dec. 2024.

\bibitem{addref2804} Y. Li, X. Zhang, T. Zeng, J. Duan, C. Wu, and D. Wu, ``Task placement and resource allocation for edge machine learning: A GNN-based multi-agent reinforcement learning paradigm," \emph{IEEE Trans. Parallel Distrib. Syst.}, vol. 34, no. 12, pp. 3073-3089, Dec. 2023.

\bibitem{ref29} E. Ranjan, S. Sanyal, and P. Talukdar, ``ASAP: Adaptive structure aware pooling for learning hierarchical graph representations," in \emph{ Proc. AAAI}, vol. 34, no. 04, pp. 5470–5477, 2020.

\bibitem{ref30} J. Schulman, F. Wolski, P. Dhariwal, A. Radford, and O. Klimov, ``Proximal policy optimization algorithms," \emph{arXiv preprint, arXiv:1707.06347}, 2017.

\bibitem{ref31} H. Mao, M. Alizadeh, I. Menache, and S. Kandula, ``Resource management with deep reinforcement learning," in \emph{Proc. HotNets}, 2016, pp. 50-56.

\bibitem{ref32} X. Zhang and J. G. Andrews, ``Downlink cellular network analysis with multi-slope path loss models," \emph{IEEE Trans. Commun.},
    vol. 63, no. 5, pp. 1881–1894, 2015.

\bibitem{ref33} J. G. Andrews, X. Zhang, G. D. Durgin, and A. K. Gupta, ``Are we approaching the fundamental limits of wireless network densification?" \emph{IEEE Commun. Mag.}, vol. 54, no. 10, pp. 184–190, 2016.


\bibitem{ref34} N. Naderializadeh and A. S. Avestimehr, ``ITLinQ: A new approach for spectrum sharing in device-to-device communication systems," \emph{ IEEE J. Sel. Areas Commun.}, vol. 32, no. 6, pp. 1139–1151, 2014.

\bibitem{ref35} Q. Shi, M. Razaviyayn, Z.-Q. Luo, and C. He, ``An iteratively weighted MMSE approach to distributed sum-utility maximization for a MIMO interfering broadcast channel," \emph{IEEE Trans. Signal Process.}, vol. 59, no. 9, pp. 4331–4340, Sep. 2011.

\bibitem{refTD3} S. Fujimoto, H. Hoof, and D. Meger, ``Addressing function approximation error in actor-critic methods," in \emph{Proc. Int. Conf. Mach. Learn. (ICML)}, 2018, pp. 1587-1596.

\end{thebibliography}
\end{document}